\documentclass[onecolumn,twoside]{IEEEtran}
\IEEEoverridecommandlockouts

\usepackage{ifpdf}
\usepackage{cite}
\usepackage[cmex10]{amsmath}
\usepackage{enumerate}
\usepackage{amsthm}
\usepackage{latexsym}
\usepackage{amssymb}
\usepackage{xspace}
\usepackage{amscd}
\usepackage{enumerate}
\usepackage{graphicx}
\usepackage{epic}

\newcommand{\N}{{\mathbb N}}

\newcommand{\Z}{{\mathbb Z}}
\newcommand{\F}{\mathbb F}
\renewcommand{\b}[1]{\mathbf{#1}}

\newcommand{\T}{{\mathcal T}}
\renewcommand{\S}{{\mathcal S}}
\newcommand{\M}{{\mathcal M}}
\renewcommand{\H}{{\mathcal H}}
\newcommand{\Q}{{\mathcal Q}}
\newcommand{\D}{{\mathcal D}}

\newcommand{\lcm}{\mathrm{lcm}}

\newcommand{\tq}{\; : \;}
\newcommand{\doble}[2]{\genfrac{}{}{0cm}{2}{#1}{#2}} 

\newcommand{\pega}[1]{\hspace{-.#1 cm}}

\usepackage{algorithmic}
\usepackage{array}

\usepackage{mdwmath}
\usepackage{mdwtab}
\usepackage{fixltx2e}

\usepackage{stfloats}

\newtheorem{theorem}{Theorem}
\newtheorem{lemma}[theorem]{Lemma}
\newtheorem{proposition}[theorem]{Proposition}
\newtheorem{definition}[theorem]{Definition}
\newtheorem{corollary}[theorem]{Corollary}
\newtheorem{remark}[theorem]{Remark}

\newtheorem{notation}[theorem]{Notation}
\newtheorem{example}[theorem]{Example}
\newtheorem{examples}[theorem]{Examples}

\begin{document}

\title{Apparent distance and a notion of BCH multivariate codes
\thanks{This work was partially supported by MINECO (Ministerio de Econom\'{\i}a
y Competitividad), (Fondo Europeo de Desarrollo Regional)
project MTM2012-35240, Programa Hispano Brasile\~{n}o de Cooperaci\'{o}n Universitaria PHB2012-0135, and Fundaci\'{o}n S\'{e}neca of Murcia. The second author has been supported by Departamento Administrativo de Ciencia, Tecnolog\'{\i}a e Innovaci\'on de la Rep\'ublica de Colombia}}

% \author{\IEEEauthorblockN{Diana H. Bueno-Carre\~no$^\star$}
% \IEEEauthorblockA{Departamento de CN y Matemáticas\\
% Pontificia Universidad Javeriana-Cali\\
% Cali, Colombia\\
% Email: dhbueno@javerianacali.edu.co}
% \and
% \IEEEauthorblockN{Jos\'e Joaqu\'{i}n Bernal \\ and \\ Juan Jacobo Sim\'on}
% \IEEEauthorblockA{Departamento de Matemáticas\\
% Universidad de Murcia\\
% 30100 Murcia, Spain.\\
% Email: josejoaquin.bernal@um.es and jsimon@um.es} }


\author{\IEEEauthorblockN{Jos\'e Joaqu\'{i}n Bernal\IEEEauthorrefmark{1},
Diana H. Bueno-Carre\~no\IEEEauthorrefmark{2} and
Juan Jacobo Sim\'on\IEEEauthorrefmark{1}. \\
\IEEEauthorblockA{\IEEEauthorrefmark{1}Departamento de Matem\'{a}ticas\\
Universidad de Murcia,
30100 Murcia, Spain.\\ Email: \{josejoaquin.bernal, jsimon\}@um.es} \\
\IEEEauthorblockA{\IEEEauthorrefmark{2}Departamento de Ciencias Naturales y Matem\'{a}ticas\\
Pontificia Universidad Javeriana, 
 Cali, Colombia\\
 Email: dhbueno@javerianacali.edu.co}
}}

\maketitle

\begin{abstract}
This paper is devoted to study two main problems: on the one hand, to compute the apparent distance of an abelian code and on the other hand, to give a notion of BCH multivariate code. To do this, we present an algorithm to compute the apparent distance of an abelian code, based on some manipulations of hypermatrices associated to its generating idempotent.  Our method uses less computations than those given in \cite{Camion} and \cite{Evans}; furthermore, in the bivariate case, the order of the computations is reduced from exponential to linear.  Then we use our techniques to develop a notion of BCH code in the multivariate case and we extend most of the classical results on BCH codes. Finally, we apply our method in two directions: we construct abelian codes from cyclic codes,  multiplying their dimension and preserving their apparent distance; and we design abelian codes with maximum dimension with respect to a fixed apparent distance and a fixed length.
\end{abstract}

% \fi
\IEEEpeerreviewmaketitle

\section{Introduction}

The oldest lower bound for the minimum distance of a cyclic code is the BCH bound (see \cite[p. 151]{HP}). Its study and its generalizations are classical topics, which include the study of the very well-known family of BCH codes. In 1970, P. Camion \cite{Camion} extended the notion of the BCH bound to the family of abelian codes by introducing the apparent distance of an abelian code. In the case of cyclic codes, the apparent distance and the lower BCH bound coincide (see paragraph below Definition~\ref{apparentdistance}).

The computation of the apparent distance of an abelian code $C$ in a semisimple ring is based, in turn, in the computation of the apparent distance of some polynomials, which correspond to all sets of idempotents belonging to $C$. This implies that an enormous number of computations are involved. Then, it is of interest to simplify it. In \cite{Evans}, Sabin computes the apparent distance of a single polynomial by using matrix manipulations, in the frame of 2-D cyclic codes (abelian codes in two variables). Even the Sabin's matrix method simplifies the original one, the number of computations was not modified. So, the problem of reducing the exponential complexity is still open.

In the Camion's mentioned paper, one may see that the apparent distance of a cyclic code equals the apparent distance of a polynomial associated to the generating idempotent. There are examples that shows that in the multivariate case the equality does not hold. Then, we wonder if we can obtain the apparent distance of an abelian code by using uniquely manipulations of the hypermatrix associated to its generating idempotent.

This is the first goal of this paper. We present an algorithm to compute the apparent distance of an abelian code, based on certain manipulations of the hypermatrix (that extends the Sabin's matrix methods) associated to the generating idempotent of a given code. Our method uses less computations than the others; in fact, in the bivariate case it has \textit{linear} complexity, in certain sense, instead of the exponential complexity of the original computation (see Remark~\ref{complejidad}); moreover, we use our techniques to develop a notion of BCH code in the multivariate case and we extend most of the classical results in BCH codes. Finally, we apply our techniques in two directions. The first one consists of constructing abelian codes from cyclic codes,  multiplying their dimension and preserving their apparent distance. The second one consists of the design of abelian codes with maximum dimension with respect to a fixed apparent distance and length.

 In Section III, we present some technical results that we use to compute the minimum apparent distance of a code by using a subset of hypermatrices. Here, we give the notion of minimum apparent distance of a hypermatrix and we show that the apparent distance of an abelian code equals the minimum apparent distance of some hypermatrices associated to its generating idempotent. Section IV is devoted to develop  an algorithm to computing the minimum apparent distance. Even most of those results are technical, they enclose techniques that will be used to give a notion of BCH multivariate code. This is done in Section V, where we also study the extension of some classical results about this BCH codes. In Section VI we give some applications.

\section{Notation and preliminaries}\label{preliminares}

In this section, we shall introduce all notation and terminology needed to understand our results. We also recall some basic facts and definitions.

All throughout, $\F_q$ denotes the field with $q$ elements where $q$ is a power of a prime $p$.
An abelian code is an ideal of a group algebra $\F_q G$, where $G$ is an abelian group. It is well-known that a decomposition $G\simeq C_{r_1}\times\cdots\times C_{r_s}$, with $C_{r_k}$ the cyclic group of order $r_k$ for $k=1,\dots,s$, induces a canonical isomorphism of $\F_q$-algebras from $\F_q G$ to 
   $$\F_q[x_1,\dots,x_s]/\left\langle x_1^{r_1}-1,\dots,x_s^{r_s}-1\right\rangle.$$ 
We denote this quotient algebra by $A_q(r_1,\dots,r_s)$. We identify the codewords with polynomials $f(x_1,\dots,x_s)$ such that every monomial satisfies that the degree of the indeterminate $x_k$ belongs to $\Z_{r_k}$, the ring of integers modulo $r_k$, that we always write as canonical representatives (that is, non negative integers less than $r_k$). We denote by $I$ the set $\Z_{r_1}\times\cdots\times \Z_{r_s}$ and we  write the elements $f \in  A_q(r_1,\dots,r_s)$ as $f=f(x_1,\dots,x_s)=\sum a_\b{i} X^\b{i}$, where $\b{i}=(i_1,\dots, i_s)\in I$ and $X^\b{i}=x_1^{i_1}\cdots x_s^{i_s}$. In case we first consider a polynomial $f \in \F_q[x_1,\dots,x_s]$, possibly having a monomial whose degree in the $k$-th variable is greater than or equal to $r_k$, then we denote by $\overline{f}$ its image under the canonical projection onto $A_q(r_1,\dots,r_s)$. We deal with abelian codes in the semisimple case; that is, we always assume that $\gcd(r_k,q)=1$ for every $k=1,\dots,s$.
  
We denote by $U_{r_i}$ the set of all $r_i$-th primitive roots of unity, for each $i=1,\dots, s$ and we define $U=\{(\alpha_1, \dots ,\alpha_s) \tq \alpha_i \in U_{r_i}\}$. If $\F_{q^v}|\F_q$ is an extension field containing every $U_{r_i}$, with $i=1,\dots,s$, it is well known that every abelian code $C$ in $A_q(r_1,\dots,r_s)$ is totally determined by its \textit{root set} $Z(C)=\left\{\alpha\in \F_{q^v}^s: f(\alpha)=0 \text{ for all } f\in C \;\text{ and }\;\alpha^{(r_1,\dots,r_s)}=1\right\}$. Fixed $\alpha=(\alpha_1,\dots,\alpha_s) \in U$, the code $C$ is  determined by its \textit{defining set}, with respect to $\alpha$, which is defined as $\D_\alpha\left(C\right)$ = $\left\{ (a_1,\dots,a_s)\in I \tq f(\alpha_1^{a_1},\dots,\alpha_s^{a_s})=0\text{ for all } f\in C\right\}.$

We recall that for positive integers $b,t, r$, the $q^t$-cyclotomic coset of $b$ modulo $r$ is the set $C_{q^t}(b) = \{a\cdot q^{ti}\in \Z_{r} \tq i\in \N\}$ (the parameter $r$ will be omitted because it will always be clear from the context). Its multivariate version is the notion of $q^t$-orbits (see, for example, \cite{BS}). Given an element $a=(a_1,\dots,a_s)\in I$, we define its \textit{$q^t$-orbit} modulo  $\left(r_1,\ldots,r_s \right)$ as
 \[  Q_t(a)=\left\{\left(a_1\cdot q^i ,\dots, a_s\cdot q^i  \right) \in I \tq i\in \N\right\},\]
 in the case $t=1$ we only write $Q(a)$.  We also recall that the multiplicative order of $a$, modulo $b$ is the first positive integer $m$, such that $b$ divides $a^m-1$. We shall denote it by $\mathcal O_b(a)$.
 
 It is easy to see that, in the semisimple case, for a fixed $\alpha \in U$ and for every abelian code $C$ in $ A_q(r_1,\dots,r_s)$, $\D_\alpha\left(C\right)$ is a disjoint union of $q$-orbits modulo $(r_1,\dots,r_s)$. Conversely, every union of $q$-orbits modulo $(r_1,\dots,r_s)$ determines an abelian code (an ideal) in $ A_q(r_1,\dots,r_s)$  (see, for example, \cite{BS}). 

To define and compute the apparent distance of an abelian code, we will associate to its defining set, with respect to $\alpha \in U$, certain hypermatrix that we will call $q$-orbits hypermatrix. For any $\b{i}\in I$ we write its $k$-th coordinate as $\b{i}(k)$. A hypermatrix, with entries in a set $R$, indexed by $I$ (or an $I$-hypermatrix over $R$) is an $s$-dimensional $I$-array, that we denote by $M=\left(a_{\b{i}}\right)_{\b{i}\in I}$, with $a_{\b{i}}\in R$ \cite{Yamada}. The set of indices, the dimension and the ground field will be omitted if they are clear by the context. In the case $s=2$ we will say that $M$ is a matrix, and for $s=1$ we will call $M$ a vector. We write $M=0$ when all of its entries are $0$; otherwise we may write $M\neq 0$.  As usual, a hypercolumn is defined as $H_M(k,b)=\left\{ a_{\b{i}}\in M\tq \b{i}(k)=b\right\}$, with $1\leq k\leq s$ and $0\leq b<r_k$, where the expression $a_{\b{i}}\in M$ means that $a_{\b{i}}$ is an entry of $M$. A hypercolumn will
  be seen as an $(s-1)$-dimensional hypermatrix. Conversely, note that given $k \in \{1, \dots , s\}$ and $b \in \Z_{r_k}$, a hypermatrix indexed by $\displaystyle\prod_{\doble{j=1}{j\neq k}}^s \Z_{r_j}$ may be viewed as a hypercolum (of certain $s$-dimensional hypermatrix) indexed by $I(k,b)=\{ i \in I \tq \b{i}(k)=b\}$. In the case $s=2$, we refer to hypercolumns as rows or columns and when $s=1$ we only say entries.

Let $D\subseteq I$. The hypermatrix afforded by $D$ is defined as $M=\left(a_{\b{i}}\right)_{\b{i}\in I}$, where $a_{\b{i}}=1$ if $\b{i}\not\in D$ and $a_{\b{i}}=0$ otherwise. When $D$ is an union of $q^t$-orbits we will say that $M$ is the hypermatrix of $q^t$-orbits afforded by $D$, and it will be denoted by $M=M(D)$. For any $I$-hypermatrix $M$ with entries in a ring, we define the support of $M$ as the set $supp(M)=\left\{ \b{i}\in I \tq a_{\b{i}}\neq 0\right\}$, whose complement will be denoted by $\D(M)$. Note that, if $D$ is a union of $q^t$-orbits, then the $q^t$-orbits hypermatrix afforded by $D$ verifies that  $\D(M(D))=D$.

Let $\Q_t$ be the set of all $q^t$-orbits in $I$, for some $t\in \N$. We define a partial ordering over the set of $q^t$-orbits hypermatrices  $\left\{M(D)\tq D= \cup Q,\text{ for some }Q\subseteq \Q_t\right\}$ as follows:

\begin{equation}\label{matrixordering}
M(D)\leq M(D') \Leftrightarrow supp\left(M(D)\right)\subseteq supp\left(M(D')\right).
\end{equation}
Clearly, this condition is equivalent to $D'\subseteq D$.

Let $\F_{q^v}|\&F_q$ be an extension field such that $U\subset \F_{q^v}^s$. The (discrete) Fourier transform of a polynomial $f\in A_q(r_1,\dots,r_s)$ (also called Mattson-Solomon polynomial \cite{Evans}), with respect to $\alpha \in U$, that we denote by $\varphi_{\alpha,f}$, is the polynomial  $\varphi_{\alpha,f}(X)=\sum_{\b{j}\in I} f(\alpha^{\b{j}})X^{\b{j}}.$  Clearly, $\varphi_{\alpha,f}\in A_{q^v}(r_1,\dots,r_s)$ and it is known that the function Fourier transform may be viewed as an isomorphism of algebras $\varphi_{\alpha}:A_{q^v}(r_1,\dots,r_s)\longrightarrow (\F_{q^v}^{|I|},\star)$, where the multiplication ``$\star$'' in $\F_{q^v}^{|I|}$ is defined coordinatewise. So, we may see $\varphi_{\alpha,f}$ as a vector in $\F_{q^v}^{|I|}$ or as a polynomial in $A_{q^v}(r_1,\dots,r_s)$. See \cite[Section 2.2]{Camion} for details.

\section{Apparent distance and BCH bounds}\label{seccion distancia aparente}

In this section we present some technical results, some of them well-known, that we use to compute the minimum apparent distance of a code.

Throughout this section $s,q$ and $r_1,\dots,r_s$ will be positive integers, with $q$ a power of a prime number $p$, such that $p\nmid r_i$, for $i=1,\dots s$. We set $I=\prod_{j=1}^s\Z_{r_j}$ and $X^\b{i}=x_1^{\b{i}(1)}\dots x_s^{\b{i}(s)}$.  We begin with the definition of apparent distance of polynomials and hypermatrices.

\subsection{The apparent distance of polynomials and matrices}

Let $f=\sum_{\b{i}\in I}a_{\b{i}}X^{\b{i}}$ be a polynomial in $ A_{q^v}(r_1,\dots,r_s)$. It is known that for each $k\in\{1,\dots,s\}$, viewed $f$ as a polynomial in $\F_{q^v}[x_1,\dots,x_s]$ it may be written as an element in $ R[x_k]$, where $R = \F_{q^v}[x_1,\dots,x_{k-1},x_{k+1},\dots,x_s]$. In this case, we set $f=f_k=\sum_{b=0}^{r_k-1}f_{k,b}x_k^b$, where $f_{k,b}=\sum_{\doble{\b{i}\in I}{\b{i}(k)=b}}a_{\b{i}}Y^{\b{i}}_k$, and $Y^{\b{i}}_k=X^{\b{i}}/x^{b}_k$. The degree of $f_k$ as polynomial in $ R[x_k]$, denoted by $\deg(f_k)$, as usual, is called the $k$-th degree of $f$. For any $\b{h}\in I$ we denote by $d_k[\b{h}]=d_k[\b{h}](f)$ the $k$-th degree of $\overline{X^{\b{h}}f}$ and by $c_k[\b{h}]=c_k[\b{h}](f)$ the coefficient of $x_k^{d_k[\b{h}]}$; i.e. $(\overline{X^{\b{h}}f})_{k,d_k[\b{h}]}$. In the case $s=1$ we denote by $d[h]=d[h](f)$ the degree of the polynomial $\overline{x^hf}$ and by $c[h]=c[h](f)$ its leading coefficient.

\begin{definition} \cite[p. 21]{Camion}. Let $s,q$, $r_1,\dots,r_s$ and $I$ be as above. Let $f\in A_{q^v}(r_1,\dots,r_s)$. The apparent distance of $f$, denoted by $d^\ast(f)$, is
\begin{enumerate}
 \item $d^\ast(0)=0$.
\item In the case $s=1$ (with $r=r_1$)
\[d^\ast(f)=\max\{r-d[h]\tq 0\leq h\leq r-1\}.\]
\item For $s\geq 2$,
$$d^\ast(f)=\max _{\mathbf h \in I}\left\{\max_{1 \leq k \leq s}  \left\{ d^\ast (c_k[\mathbf h])(r_k-d_k[\mathbf h])\right\}\right\}.$$ 
\end{enumerate}
\end{definition}

\begin{example}\rm{
Set $f=x_2^3-(x_1+1)x_2$ in $A_3(2,4)$ and take $\b{h}=(1,2)$. Then $\overline{X^\b{h} f}=x_1x_2^2f= (x_2-x_2^3)x_1-x_2^3=(-1-x_1)x_2^3+x_1x_2$. In this case, $c_1[(1,2)]=x_2-x_2^3$, $d_1[(1,2)]=1$, $c_2[(1,2)]=-1-x_1$, $d_2[(1,2)]=3$ and we get $d^*(c_1[(1,2)])=2$ and $d^*(c_2[(1,2)])=1$. So that, for $\b h=(1,2)$, we have that  $\max_{1 \leq k \leq 2}  \left\{ d^\ast ( c_k[\b h])(r_k-d_k[\b h])\right\}=2$. One may check that considering all the elements in $I=\Z_2\times\Z_4$ we obtain $d^\ast(f)=4$.
}\end{example}

In \cite[Section 2.3]{Evans}, Sabin computes the apparent distance by using matricial methods for polynomials in two variables. As a generalization of those techniques, we introduce the notion of apparent distance of a hypermatrix.

For a positive integer $r$, we say that a list of canonical representatives  $b_0,\dots,b_l$ in $\Z_r$ is a list of consecutive integers modulo $r$, if for each $0\leq k<l$ we have that $b_{k+1}\equiv b_{k}+1\mod r$. If $b=b_k$ (resp. $b=b_{k+1}$) we denote $b^+=b_{k+1}$ (resp. $b^-=b_k$).

\begin{definition}\label{hipercolumnas cero adyacentes}
 Let $s,q$, $r_1,\dots,r_s$ and $I$ be as above. Let $M$ be a hypermatrix, $k\in \{1,\dots,s\}$, $b\in \Z_{r_k}$ and $H_M(k,b)$ a nonzero hypercolumn. The set of zero hypercolumns adjacents to $H_M(k,b)$ is the set of hypercolumns
	\[CH_M(k,b)=\{H_M(k,b_0), H_M(k,b_1), \dots , H_M(k,b_l)\}\]
such that $H_M(k,b_j)=0$ for all $j\in \{0,\dots, l\}$,  $b_0,\dots,b_l$ is a list of consecutive integers modulo $r_k$, $b^+=b_0$ and $H_M(k,b_l^+)\neq 0$.

In the case $s=1$ we replace hypercolumns by entries.
\end{definition}

\begin{notation}
 We denote by $\omega_M(k,b)$ the value $|CH_M(k,b)|$; in the case $s=1$ we write $\omega_M(b)=\omega_M(1,b)$.
\end{notation}

We point out that for some values $k$ and $b$ it may happen that $\omega_M(k,b)=0$.

\begin{definition}
 Let $s,q$, $r_1,\dots,r_s$ and $I$ be as above. Let $M$ be a hypermatrix over $\F_q$ and $k\in \{1,\dots,s\}$. 
\begin{enumerate}

 \item If $M$ is the  zero hypermatrix, its apparent distance is $d^\ast(0)=0$.

\item In case $s=1$, the apparent distance of a vector $M$ is $d^\ast(M)=\max_{b\in \Z_r}\{\omega_M(b)+1\}$.

\item For $s\geq 2$, we give the definition in two steps:\\
\textit{ (3.1)} The apparent distance of $M$ with respect to the $k$-th variable is $$d_k^* (M) =\max_{b\in \Z_{r_k}}\left\{(\omega_M(k,b)+1)\cdot d^\ast (H_M(k,b)) \right\}.$$
Then

\textit{ (3.2)} the apparent distance of $M$ is $d^\ast (M) =\max_{1 \leq k \leq s}\{ d_k^*(M) \}.$
\end{enumerate}
\end{definition}

As the apparent distance of a hypermatrix is a maximum we focus on those hypercolumns involved in the computation of the maximum by the following definition.

\begin{definition} \label{def involved hypercolumns}
Let $M$ be a nonzero hypermatrix. We say that a pair $(k,b)$, where $k\in \{1,\dots,s\}$ and $b\in \Z_{r_k}$, is an involved pair (in the computation of $d^\ast (M)$) if $d^\ast (M)=(\omega_M(k,b)+1)d^\ast (H_M(k,b)) $. The hypercolumn $H_M(k,b)$ is called, in turn, an involved hypercolumn (in the computation of $d^\ast(M)$).

We denote the set of involved pairs by $Ip(M)$.
\end{definition}

\begin{examples}\label{ejemplos dis apar}\rm{
 \textit{1)}  First, we consider $I=\Z_4$ and set $M=(2,0,0,1)$ over $\F_3$. In this case,  $d^\ast (M)=\max\{ \omega_M(b)+1 \tq 0 \leq b \leq 3 \}=\max\{ \omega_M(0)+1, \omega_M(3)+1\}=\max \{3,1\}=3$ and $Ip(M)=\{0\}$.

 \noindent \textit{2)} 
 Now set $I=\Z_{3} \times \Z_5$ and consider the matrix over $\F_2$,
$$M= \left(
    \begin{array}{ccccc}
     1 & 0 & 0 & 0 & 0  \\
     1&  1 & 0 & 0& 1  \\
     1 & 1 & 0 & 0 & 1  \\
          \end{array}
  \right).$$ 
 We begin by computing $d_1^\ast (M)$. In this case
 \[\begin{array}{|c|c|c|c|}
 \hline
      b & \omega_{M}(1,b) & d^*H_{M}(1,b) & (\omega_{M}(1,b)+1)\cdot d^*H_{M}(1,b)  \\ \hline
      0 & 0 & 5 & 5\\ \hline
      1 & 0 & 3 & 3\\ \hline
      2 & 0 & 3 & 3\\ \hline
\end{array}\]
so $d_1^\ast (M)=5 $. Now, with respect to $x_2$ we get 
$$
\begin{array}{|c|c|c|c|}
 \hline
      b & \omega_{M}(2,b) & d^*H_{M}(2,b) & (\omega_{M}(2,b)+1)\cdot d^*H_{M}(2,b)  \\ \hline
      0 & 0 & 1 & 1\\ \hline
      1 & 2 & 2 & 6\\ \hline
      4 & 0 & 2 & 2\\ \hline
\end{array}
$$ 
Therefore, $d_2^\ast (M) =6$ and hence $d^\ast (M) =6$. We also have that $Ip=\{(2,1)\}$.

\noindent \textit{3)}  Set $q=2$, $r_1=3$, $r_2=3$ and $r_3=5$. Let $M=M(D)$ be the hypermatrix afforded by the set of $2$-orbits 
  \begin{eqnarray*}
D&=& Q(0,0,0)\cup Q(1,0,0)\cup Q(0,1,0)\cup Q(0,0,1)\\
&& \cup Q(1,2,0)\cup Q(1,2,1)\cup Q(1,2,2)\cup Q(1,1,0)\\ 
&&\cup Q(0,1,1)\cup Q(1,0,2)\cup Q(0,1,2).
\end{eqnarray*}

In the following tables we show all computations. We recall that we only have to compute apparent distances on nonzero hypercolumns.

$$
\begin{array}{|c|c|c|c|}
 \hline
      b & \omega_{M}(1,b) & d^*H_{M}(1,b) & (\omega_{M}(1,b)+1)\cdot d^*H_{M}(1,b)  \\ \hline
      1 & 0 & 4 & 4\\ \hline
      2 & 1 & 8 & 16\\ \hline
      
\end{array}
$$ 

$$
\begin{array}{|c|c|c|c|}
 \hline
      b & \omega_{M}(2,b) & d^*H_{M}(2,b) & (\omega_{M}(2,b)+1)\cdot d^*H_{M}(2,b)  \\ \hline
      0 & 0 & 8 & 8\\ \hline
      1 & 0 & 6 & 6\\ \hline
      2 & 0 & 6 & 6\\ \hline
      
\end{array}
$$ 

$$
\begin{array}{|c|c|c|c|}
 \hline
      b & \omega_{M}(3,b) & d^*H_{M}(3,b) & (\omega_{M}(3,b)+1)\cdot d^*H_{M}(3,b)  \\ \hline
      1 & 0 & 6 & 6\\ \hline
      2 & 0 & 6 & 6\\ \hline
      3 & 0 & 6 & 6\\ \hline
      4 & 1 & 6 & 12\\ \hline     
\end{array}
$$

So that  $d_1^\ast (M)=16$, $ d_2^\ast (M)=6$ and $ d_3^\ast (M)= 6$.
Hence $d^\ast (M)=16$ and $Ip(M)=\{(1,2)\}$.

}\end{examples}

Now we present the relationship between the apparent distance of hypermatrices and polynomials. Let $s,q$, $r_1,\dots,r_s$, $I$ and $A_q(r_1,\dots,r_s)$ be as above. We recall that for any polynomial $f=\sum_{\b{i}\in I}a_{\b{i}}X^\b{i}\in A_q(r_1,\dots,r_s)$ the hypermatrix of coefficients of $f$ is $M(f)=\left(a_{\b{i}}\right)_{\b{i}\in I}$. Moreover, for any $k\in \{1,\dots,s\}$ writing $f=f_k=\sum_{b=0}^{r_k-1}f_{k,b}x_k^b$ one may check that $M(f_{k,b})=H_M(k,b)$. The following theorem, that generalizes the arguments in \cite[pp. 189-190]{Evans}, relates the apparent distances of $f$ and $M(f)$.

\begin{theorem}
 Let $s,q$, $r_1,\dots,r_s$, $I$, and $A_q(r_1,\dots,r_s)$ be as above. For any polynomial $f\in A_q(r_1,\dots,r_s)$ with coefficient hypermatrix $M(f)$, the equality $d^\ast(f)=d^\ast(M(f))$ holds.
\end{theorem}
\begin{proof}
We proceed by induction on $s$. Recall that a hypercolumn may be viewed as an $(s-1)$-dimensional hypermatrix.  In case $s=1$, we consider a polynomial $f=\sum_{i=0}^{r-1}a_ix^i$ and take its matrix of coeffcients $M=M(f)=(a_0 \dots a_{r-1})$. Take any $i\in supp(f)$ and set $h=r-1-\omega_M(i)-i$. Since $a_{i+j}=0$ for all $0 < j\leq \omega_M(i)$ then $d[h]=\deg(\overline{x^hf})=i+h$ and its leading coefficient is $c[h]=a_i$. Hence $\omega_M(i)=r-1-d[h]$ which give us $d^\ast(M(f))\leq d^\ast (f)$. Now, for any $h\in I=\Z_r$, we have that $\deg(\overline{x^hf})=d[h]$ then $\omega_M(d[h]-h)\geq r-1-d[h]$, and hence $d^\ast (f)\leq d^\ast(M(f))$.

Suppose that the result is true for every integer $1\leq t < s$. We will prove that it is also true for $s$. Consider $f=\sum_{\b{i}\in I}a_{\b{i}}X^{\b{i}}$ and $M=M(f)$. If we write $f=f_k=\sum_{b=0}^{r_k-1}f_{k,b}x_k^b$ then $M(f_{k,b})=H_M(k,b)$ for all $k \in \{1,\dots, s\}$ and $b \in \Z_r$. Suppose that $H_M(k,b)\neq 0$ for some $b\in\{0,\dots,r_k-1\}$ and $k\in\{1,\dots,s\}$, and consider $\b{h}\in I$ such that $\b{h}(k)=r_k-1-\omega_M(k,b)-b$ and $\b{h}(k')=0$ for all $k'\neq k$. Then $d_k[\b{h}]=b+\b{h}(k)$ and $c_k[\b{h}]=f_{k,b}$.

So, $r_k-d_k[\b{h}]=\omega_M(k,b)+1$ and $d^\ast(f_{k,b})=d^\ast(H_M(k,b))$ by induction hypothesis. This implies that $d^\ast(M)\leq d^\ast (f)$.

Now take $\b{h}\in I$, $k\in\{1,\dots,s\}$ and consider $\overline{X^{\b{h}}f}$, $d_k[\b{h}]$ and $c_k[\b{h}]$. In this case, $\omega_M\left(k,d_k[\b{h}]-\b{h}(k)\right)\geq r_k-1-d_k[\b{h}]$ and hence $d^\ast(f)\leq d^\ast (M)$.
\end{proof}

\subsection{The apparent distance of an abelian code}

Now we are going to define the apparent distance of an abelian code. We recall that as $A_q(r_1,\dots,r_s)$ is a semisimple ring, every ideal is generated by an idempotent which decompose as sum of primitive idempotents. Having in mind this fact, it is easy to see that the following definition is equivalent to Camion's definition in \cite{Camion}.

\begin{definition}\label{apparentdistance}
 Let $C$ be a code in $A_q(r_1,\dots,r_s)$. The apparent distance of $C$, with respect to $\alpha \in U$, is $d_{\alpha}^* (C)= \min\left\{d^\ast \left(M(\varphi_{\alpha,e})\right) \tq 0\neq e^2 = e \in C\right\},$ where $\varphi_{\alpha,e}$ denotes the image of $e$ under the discrete Fourier transform, with respect to $\alpha$, as we denoted in the previous section. The apparent distance of $C$ is $d^\ast (C)= \max\left\{d_{\beta}^\ast (C) \tq \beta \in U \right\}$. We also define the set of optimized roots of $C$ as $\mathcal R(C)=\{\beta \in U \tq d^\ast (C)=d_{\beta}^\ast (C))\}.$
\end{definition}

In \cite{Camion} one may see that $d_\alpha^\ast(C)=\min\{d_\alpha^\ast(\varphi_{\alpha,f}:f\in C\}$ and that $\omega(f)\geq d_\alpha^\ast\left(\varphi_{\alpha,f}\right)$ for all $f\in C$. These facts imply that the apparent distance is a lower bound for the minimum distance of any abelian code; in fact, the apparent distance of any cyclic code is exactly the maximum of all its BCH bounds (what P. Camion calls the BCH bound of a cyclic code) \cite[pp. 21-22]{Camion}.

\begin{theorem}[Camion\cite{Camion}]
 For any abelian code $C$ in $A_q(r_1,\dots,r_s)$ the inequality $d^\ast(C)\leq d(C)$ holds. 
\end{theorem}

Note that, if $e\in A_{q}(r_1,\dots,r_s)$ is an idempotent and $E$ is the ideal generated by $e$ then for any $\alpha \in U$ we have that $\varphi_{\alpha, e}\star\varphi_{\alpha, e}=\varphi_{\alpha, e}$.  Then if $\varphi_{\alpha, e}=\sum_{\b{i}\in I} a_{\b{i}}X^{\b{i}}$ we have that $a_{\b{i}}\in \{1,0\}\subseteq \F_q$ and $a_{\b{i}}=0$ if and only if $\b{i}\in \D(E)$. So that $M(\varphi_{\alpha,e})=M(\D_{\alpha}(E))$. Conversely, a $q$-orbits hypermatrix afforded by a set $D$ which is union of $q$-orbits corresponds with the image, under the Fourier transform with respect to some $\alpha \in U$, of an idempotent; to witt, the generating idempotent of the ideal in $A_q(r_1,\dots,r_s)$ determined by $D$.

Now let $C$ be an abelian code, $\alpha \in U$ and $M$ the hypermatrix aforded by $\D_{\alpha}(C)$. For any $q$-orbits hypermatrix $P\leq M$ [see (\ref{matrixordering})] there exists an idempotent $e'\in C$ such that $P=M(\varphi_{\alpha, e'})$. So we may conclude that the apparent distance of an abelian code may be computed by means of $q$-orbits hypermatrices $P\leq M(\varphi_{\alpha, e})$; that is
\begin{eqnarray*}
  \min\{d^\ast(P)\tq 0\neq P\leq M\}= \\ \min\{d^\ast(M(\varphi_{\alpha, e}))\tq
  0\neq e^2=e\in C\}= d_\alpha^\ast(C).  
\end{eqnarray*}
 This fact drives us to give the following definition.
\begin{definition}
 In the setting described above, for a $q^t$-orbits hypermatrix $M$, its minimum apparent distance is
	\[mad(M)=\min\{d^\ast(P)\tq 0\neq P\leq M\}.\]
\end{definition}

Finally, in the next theorem we set the relationship between the apparent distance of an abelian code and the minimum apparent distance of the coefficient hypermatrices of the Fourier transforms of its generating idempotent.

\begin{theorem}
Let $C$ be an abelian code in $A_q(r_1,\dots,r_s)$ and let $e$ be its generating idempotent. Then $d_{\alpha}^\ast (C)=mad\left(M(\varphi_{\alpha,e})\right)$ ($\alpha\in U$). Therefore, $d^\ast (C)=\max\{mad\left(M(\varphi_{\alpha,e})\right): \alpha\in U\}$.
\end{theorem}

In next section we present a technique to compute the minimum apparent distance of a hypermatrix and thereby to compute the apparent distance of an abelian code.

\begin{remark}\label{distintas raices}\rm{
Let us note that to get the maximum value that defines $d^\ast(C)$ we do not need to compute the maximum apparent distance over all the elements of $U$. Indeed, let $Q(a_1),Q(a_2),\dots ,Q(a_h)$ be all different $q$-orbits modulo $(r_1,\dots , r_s)$ and fix the representatives $a_1, \dots , a_h$. Chose $\alpha \in U$ to get a defining set $\D_{\alpha}(C)$. We look for the elements $\beta=(\beta_1, \dots , \beta_s) \in U$ for which it is possible that $D_{\beta}(C)\neq D_{\alpha}(C)$. In this way, $\beta \in U$ has to satisfy $\beta^{\mathbf{a}_iq^t}=\alpha$ for some $t\in \Z$ and  $\mathbf{a}_i=(a_{i1}, \dots , a_{is})$ such that $\gcd(a_{ij},r_j)=1$ with $j=1,\dots, s$.  In this case, it is clear that $D_{\beta}(C)=\mathbf{a}_i \cdot D_{\alpha}(C)$, where the multiplication has the obvious meaning. Moreover, since $D_{\beta^{\mathbf{a}_i}}(C)= D_{(\beta_1^{a_{i1}q}, \dots , \beta_s^{a_{is}q})}(C)$ and  $|Q(\mathbf{a}_i)|=\gcd \{\mathcal O_{r_i}(q)\}_{i=1}^s$  for all $\mathbf{a}_i=(
 a_{i1}, \dots , a_{is})$ such that $\gcd(a_{ij},r_j)=1$, $j=1,\dots, s$, we have to consider at most $\frac{\prod_{i=1}^s \phi(r_i)}{\gcd\{\mathcal O_{r_i}(q)\}_{i=1}^s}$ defining sets or elements in $U$. 
}\end{remark}

Then, we denote by $K(r_1, \dots r_s)=\{a_i=(a_{i1}, \dots , a_{is})\tq \gcd(a_{ij},r_j)=1, j=1,\dots, s, i=1, \dots, h\}$ and fixed $\alpha\in U$ we define $\mathcal R_\alpha=\{\beta \in U \tq \beta^{a_i}=\alpha, a_i \in K(r_1, \dots, r_s)\}$.  So, in practice, fixed $\alpha\in U$, $d^\ast (C)= \max\left\{d_{\beta}^\ast (C) \tq \beta \in \mathcal R_\alpha \right\}$.

\section{Computing the minimum apparent distance of a hypermatrix}

Let $s,q$, $r_1,\dots,r_s$ and $I$ be as in the preceding section, and let $\Q_t$ be the set of all $q^t$-orbits in $I$, for some $t\in \N$. For an arbitrary subset $\Q' \subseteq \Q_t$ we set  $D=\displaystyle \cup_{Q\in \Q'} Q$, and construct $M=M(D)$, the $q^t$-orbits hypermatrix afforded by $D$. Consider an arbitrary hypercolumn of $M$, say $H_M(k,b)$, where $k\in \{1,\dots,s\}$ and $b\in \Z_{r_k}$. Recall that $I(k,b)=\{\mathbf i \in I: \mathbf{i}(k)=b\}$, $H_M(k,b)=\{a_{\mathbf{i}} \in M: \mathbf{i} \in I(k,b)\}$ and consider the set $D_M(k,b)=I(k,b)\setminus supp(H_M(k,b))$. 

We claim that $H_M(k,b)$ may be viewed as a $(s-1)$-dimensional hypermatrix of $q^{t'}$-orbits, where $t'=t|C_{q^t}(b)|$ and $C_{q^t}(b)$ is the $q^{t}$-cyclotomic coset of $b$, modulo $r_k$, and, as such, it is the hypermatrix afforded by $D_M(k,b)$. To prove this, first note that for each $\b{i}\in I(k,b)$,  we have that $q^{t'}\mathbf{i}(k)=\mathbf{i}(k)=b$, hence $q^{t'}\mathbf{i} \in I(k,b)$; that is, $I(k,b)$ is closed under $q^{t'}$-orbits.  Now let $\mathbf{i} \in I(k,b)$ be such that  $a_{\mathbf{i}}=0$.  Since $\mathbf{i} \in I$, there exists $Q \in \Q_t$ such that  $\mathbf{i} \in Q$ and then $q^{t'}\mathbf{i} \in Q$, which implies that $a_{q^{t'}\mathbf{i}}=0$, because $M$ is a $q^t$-orbits hypermatrix.  This shows that $H_M(k,b)$ is a $q^{t'}$-orbits hypermatrix. The fact that $\D\left(H_M(k,b)\right)=D_M(k,b)$ is obvious.

\begin{proposition}
Let $M$ be a $q^t$-orbits hypermatrix and $N<M$. Then,  for each $k=1,\dots,s$ and $b\in \Z_{r_k}$, $H_N(k,b)\leq H_M(k,b)$, viewed as $q^{t'}$-orbits hypermatrices, where $t'=t|C_{q^{t}}(b)|$ and $C_{q^t}(b)$ is the $q^{t}$-cyclotomic coset of $b$, modulo $r_k$.
\end{proposition}

\begin{proof}
Since $N<M$ then $\D(M) \subset \D(N)$; so that $D_M(k,b) \subseteq D_N(k,b)$. Having in mind the ordering in  (\ref{matrixordering}), the claim in the paragraph prior this result shows us that $H_N(k,b)\leq H_M(k,b)$ viewed as $q^{t'}$-orbits hypermatrices.
\end{proof}

\begin{lemma}\label{subhipermat soporte maximo}
 Let $M$ be a nonzero $q^t$-orbits hypermatrix. Consider $k \in \{1, \dots , s\}$ and $b \in \Z_{r_k}$.  Let $A$ be an $(s-1)$-dimensional hypermatrix indexed by $\prod_{\doble{j=1}{j\neq k}}^s \Z_{r_j}$, such that $supp(A) \subseteq supp(H_M(k,b))$.  Then, there exists $N\leq M$ such that:
\begin{enumerate}
\item $supp(H_N(k,b)) \subseteq supp(A)$.
\item If $P < M$ verifies that $supp(H_P(k,b)) \subseteq supp(A)$ then $P \leq N$.
\item If $A$ is a $q^{t|C_{q^t}(b)|}$-orbits hypermatrix where $C_{q^t}(b)$ is the $q^{t}$-cyclotomic coset of $b$, modulo $r_k$, then $A = H_N(k,b)$.
\end{enumerate}
Hence, $N$ is the $q^t$-orbits hypermatrix with maximum support (with respect to the inclusion) such that $supp(H_N(k,b)) \subseteq supp(A)$.
\end{lemma}

\begin{proof}
 Set $\bar{A}=supp(H_M(k,b)) \setminus supp(A)$ and let $N$ be the $q^t$-orbits hypermatrix such that $\D(N)=\D(M)\cup \left( \cup_{\mathbf{i} \in \bar{A}} Q_t(\mathbf{i}) \right)$. It is clear that $N \leq M$ and $supp(H_N(k,b))\subseteq supp(H_M(k,b))$. Let us see that the conditions of our lemma are satisfied.  
 
\textit{1.} Take any $\mathbf{i} \in supp(H_N(k,b))$. Then $Q_t(\mathbf{i}) \cap \D(N) = \emptyset$ and then $Q_t(\mathbf{i}) \cap \bar{A} = \emptyset$, so that $Q_t(\mathbf{i}) \subseteq supp(A)$, because $\mathbf i\in supp(H_M(k,b))$. Hence $\mathbf{i} \in supp(A)$.

\textit{2.} Suppose that $P < M$ verifies that $supp(H_P(k,b)) \subseteq supp(A)$. Take any $\mathbf{i} \in supp(P)$; that is, $\mathbf{i} \notin \D(P)$.  Then $Q_t(\mathbf{i}) \cap \D(P)= \emptyset$. We are going to see that $Q_t(\mathbf{i}) \cap \D(N)= \emptyset$. 

As $P<M$, then $\D(M) \subset \D(P)$ and so $Q_t(\mathbf{i}) \cap \D(M)= \emptyset$, then $\mathbf{i} \notin \D(M)$. So, if $Q_t(\mathbf{i}) \cap \D(N) \neq \emptyset$, then  there exists  $\mathbf j\in \bar A$ such that $Q_t(\b{i})=Q_t(\mathbf j)$. From here, we have that $\mathbf j\in supp(P)$ and since  $~\mathbf j(k)=b$ then $\mathbf j\in supp(H_P(k,b))\subseteq supp(A)$, which is impossible.

\textit{3.} Assume that $A$ is a $q^{t|C_{q^t}(b)|}$-orbits hypermatrix. Suppose that $supp(A)\setminus supp(H_N(k,b)) \neq \emptyset$ and consider $\mathbf{i} \in supp(A)\setminus supp(H_N(k,b))$. Then $Q_t(\mathbf{i}) \subseteq \D(N)$.  Let $N'<M$  be the hypermatrix afforded by $\D(N)\setminus Q_t(\mathbf i)$. Then $N <N'$ and 
$$D_{N'}(k,b)= D_N(k,b) \setminus Q_{t|C_{q^t}(b)|}(\mathbf{i}) \supseteq \D(A).$$ 
Hence,  $H_{N'}(k,b) \leq A$, which contradicts the statement \textit{(2)} which was already proved.
\end{proof}

By applying repeatedly lemma above we get the following corollary.

\begin{corollary}\label{subhipermat de soporte maximo corolario}
 Let $M$ be a nonzero $q^t$-orbits hypermatrix. Consider the list of pairs $\left(k_1,b_1\right),\dots,\left(k_l,b_l\right)$ where $1\leq k_j\leq s$ and $b_j \in \Z_{r_{k_j}}$ with $j=1,\dots, l$. Then, there exists a hypermatrix $N\leq M$ with maximum support such that $H_N\left(k_j,b_j\right)=0$, where $j=1,\dots,l$.
\end{corollary}

Our next result shows a sufficient condition to get at once the minimal apparent distance of a hypermatrix.

\begin{proposition}\label{matrizdamvarias}
Let $D$ be a union of $q^t$-orbits and $M=M(D) \neq 0$. Let $H_M(k,b)$ be an involved hypercolumn (see Definition~\ref{def involved hypercolumns}) in the computation of $d^\ast (M)$, with $1\leq k\leq s$ and $b\in \Z_{r_k}$. If $d^\ast(H_M(k,b))=1$ then $mad(M)=d^\ast (M)$.
\end{proposition}

\begin{proof}
By hypothesis, we have that $d^\ast (M)=(\omega_M(k,b)+1)\cdot d^\ast (H_M(k,b))=\omega_M(k,b)+1$. Consider a hypermatrix $0\neq M'\leq M$. Clearly, if there is $l\in \Z_{r_k}$ such that $H_M(k,l)=0$  then $H_{M'}(k,l)=0$ ; so that, as  $H_M(k,b)\neq 0$ then $CH_M(k,b) \subseteq CH_{M'}(k,b')$, for some $b'\in \Z_{r_k}$  and hence $d^\ast ( M')\geq (\omega_{M'}(k,b')+1) d^\ast (H_{M'}(k,b')) \geq \omega_{M}(k,b)+1= d^\ast (M)$.
\end{proof}

Note that the proposition above has not interest in the case $s=1$; however, if $M$ is a vector we have that $mad(M)=d^\ast (M)$.  In fact, if $P$ and $M$ are $q^t$-orbits vectors with $P<M$, then $d^\ast (P) < d^\ast (M) $ implies $P=0$. In the multivariate case we have the following result.

\begin{lemma} \label{submatricessubcolumnas}
Let $D$ be a union of $q^t$-orbits such that $M=M(D)\neq 0$. Let $H_M(k,b)$ be an involved hypercolumn in the computation of $d^\ast (M)$, with $1\leq k\leq s$ and $b\in \Z_{r_k}$.  If $P<M$ and $d^\ast (P)<d^\ast (M)$ then $d^\ast (H_P(k,b))<d^\ast (H_M(k,b))$. Consequently, $H_P(k,b)<H_M(k,b)$ as $q^{t|C_{q^t}(b)|}$-orbits hypermatrices.
\end{lemma}

\begin{proof}
As $d^\ast (P)<d^\ast (M)$, then
\begin{eqnarray*}
 d^\ast (M)&=&d^\ast (H_M(k,b))(\omega_M(k,b)+1)>\\
&>&d^\ast (P) \geq d^\ast (H_P(k,b))(\omega_P(k,b)+1). 
\end{eqnarray*}

If $H_P(k,b)=0$ then we have finished; so, suppose that $H_P(k,b)\neq 0$. Then we have that $CH_M(k,b)\subseteq CH_P(k,b)$ and so $\omega_M(k,b)\leq \omega_P(k,b)$. This fact, together with the inequality above imply that $d^\ast (H_P(k,b))<d^\ast (H_M(k,b))$, which, in turn, implies that $H_P(k,b)<H_M(k,b)$.
\end{proof}

In the rest of this section, we present our method to compute the minimum apparent distance of a hypermatrix. The proof uses recursion on the dimension of a hypermatrices; so that, for the convenience of the reader we begin by considering the case of dimension $2$; that is, matrices. The case of dimension $1$ is covered by the argument in paragraph below Proposition~\ref{matrizdamvarias} 

\begin{proposition} \label{teodam2}
Let $\Q_t$ be the set of all $q^t$-orbits modulo $(r_1,r_2)$, $\mu\in \{1,\dots, |\Q_t|-1\}$ and $\left\{Q_j\right\}_{j=1}^{\mu} $ a subset of $\Q_t$. Set $D=\cup_{j=1}^\mu Q_j$ and $M=M(D)$. Then there exist two sequences: the first one is formed by nonzero $q^t$-orbits matrices, 
$$M=M_0 > \dots > M_l\neq 0$$
and the second one is formed by positive integers
$$m_0 \geq  \dots \geq m_l $$  
with $ l\leq \mu$ and $m_i\leq d^\ast (M_i)$, for $0 \leq i \leq  l$, verifying the following property:
\begin{itemize}
 \item[( I )] If $P$ is a $q^t$-orbits matrix such that $0 \neq P \leq M$, then $d^\ast (P) \geq m_l$ and if $d^\ast (P) < m_{i-1}$ then $P \leq M_i$, where $0 < i\leq l$ .
\end{itemize}
 Moreover,  if $l' \in \{0, \dots , l\}$ is the first element satisfying  that $m_{l'}=m_l$ then $d^\ast (M_{l'})=mad(M)$.
\end{proposition}
\begin{proof}
First note that $M\neq 0$ because $\mu\leq |\Q_t|-1$. We shall construct our sequences by recursion. We shall construct two sequences by recursion satisfying condition (\,I\,). Set $M_0=M$, $m_0=d^\ast (M)$ and let $Ip(M)$ be the set of involved pairs in the computation of $d^\ast (M)$  (see Definition~\ref{def involved hypercolumns}). If there is a pair $(k,b)\in Ip(M)$, with $k\in \{1,2\}$ and $b\in \Z_{r_k}$, such that $d^\ast (H_M(k,b))=1$ then by Proposition~\ref{matrizdamvarias} we have finished (with $l=0$); so, suppose that $d^\ast (H_M(k,b))\neq 1$ for any pair $(k,b)\in Ip(M)$. In this case, by Corollary~\ref{subhipermat de soporte maximo corolario}, we may construct the $q^t$-orbits matrix, $M_1<M$ with maximum support, such that $H_M(k,b)=0$ for all $(k,b)\in Ip(M)$.

We claim that for any $q^t$-orbits matrix $P< M$, if $d^\ast (P)<m_0$ then $P\leq M_1$. Assume that $P$ is a $q^t$-orbits matrix with $P< M$ and  $d^\ast (P)<m_0$. Take any $(k,b)\in Ip(M)$. By Lemma~\ref{submatricessubcolumnas}, $d^\ast (H_{P}(k,b))<d^\ast (H_{M}(k,b))=mad(H_M(k,b))$ because $H_M(k,b)$ is a vector and then  $H_{P}(k,b)=0$ (see coment below Proposition~\ref{matrizdamvarias}). Thus $P\leq M_1$ because $M_1$ has maximum support. So, if $M_1=0$ then we have finished by taking again $l=0$.

If $M_1\neq 0$, we finish the base step by defining $m_1=\min\{m_0,d^\ast (M_1)\}$ and so we get $M_1$ and $m_1$ satisfying the required condition by the preceding paragraph.

Suppose we have constructed, for $ \delta\in \{1,\dots, \mu-1\}$ the sequences $M=M_0 > \dots > M_\delta\neq 0$ and $m_0\geq \dots \geq m_\delta$ such that for all $i\in\{1,\dots,\delta\}$, it happens that $m_i = \min\{m_{i-1},d^\ast (M_i)\}$ and if $P\leq M$ satisfies that $d^\ast (P)< m_{i-1}$ then $P < M_i$. 

To get the step $\delta+1$, we shall proceed with $M_\delta$ as we have done for $M_0$. First we check the existence of a pair $(k,b)\in Ip(M_\delta)$, with $k \in \{1,2\}$ and $b\in \Z_{r_k}$, such that $d^\ast (H_{M_\delta}(k,b))=1$. If this happens, then by Proposition~\ref{matrizdamvarias} we have finished (with $l=\delta$); so, suppose that $d^\ast (H_M(k,b))\neq 1$ for any pair $(k,b)\in Ip(M_\delta)$. As above, Corollary~\ref{subhipermat de soporte maximo corolario} allows us to construct the $q^t$-orbits matrix, $M_{\delta+1}<M_\delta$ with maximum support, such that $H_{M_{\delta+1}}(k,b)=0$ for all $(k,b)\in Ip(M_\delta)$. If $M_{\delta+1}= 0$, we finish by setting $l=\delta$; otherwise, Lemma~\ref{submatricessubcolumnas} together with the definition of minimum apparent distance show us that any $q^t$-orbits matrix $0\neq P< M$, satisfying $d^\ast (P)<m_\delta$ verifies that $P\leq M_{\delta+1}$, because $H_{M_{\delta+1}}(k,b)=0$ for all pair $(k,b)\in Ip(M_\delta)$.

We define $m_{\delta+1}=\min\{m_\delta,d^\ast (M_{\delta+1})\}$ which increases our sequences satisfying (I) by the arguments in paragraph above.

This process must stop at most in $|\Q_t|-\mu$ steps, because the supports of the considered $q^t$-orbits matrices differs in at least one $q^t$-orbit.  The sequences end at the step $l$, in which either $d^\ast (H_{M_l}(k,b))= 1$ for some pair $(k,b)\in Ip(M_l)$, with $k\in \{1,2\}$, or $M_{l+1}=0$. We note that, if $l'\in \{1,\dots,l\}$ is the first element such that $m_l=m_{l'}$ then $m_{l'}=d^\ast(M_{l'})=d^\ast(M_l)$. We shall prove that 
$mad(M)=d^\ast (M_{l'})$. Suppose that  there exists a $q^t$-orbits matrix 
$0\neq P \leq M$ with $d^\ast (P)< m_{l'}=m_l$. By the construction of our sequences one must have that $P\leq M_{l+1}\leq M_l$. Now, if the sequence of matrices stops because $d^\ast (H_{M_l}(k,b))= 1$ for some pair $(k,b)\in Ip(M_l)$, then by Proposition \ref{matrizdamvarias} we have that $mad(M_l)=d^\ast(M_l)$ and then $P=0$. Now if the sequence stops because $M_{l+1}=0$ then $P=0$. So in both cases it happens that $P=0$, which is impossible. Hence $mad(M)=d^\ast(M_{l'})$. 
\end{proof}

\begin{remark}\label{complejidad}\rm{
 Let us comment briefly the complexity of our algorithm above in the particular case $s=2$. Given an abelian code $C$ in $A_q(r_1,r_2)$, we denote by $\{e_1,\dots,e_\mu\}$ the set of primitive central idempotents which belong to $C$ and let $\mathcal Q$ be the set of all $q$-orbits modulo $(r_1,r_2)$. The computation of the apparent distance of $C$, with respect to any $\alpha \in U$, by the methods given in \cite{Camion} and \cite{Evans} needs $2^{\mu}-1$ computations of apparent distances of $q$-orbits matrices, while the number of computations in our method is at most $\mu$; that is we change exponential complexity by linear complexity.
}\end{remark}

\begin{example}\label{ejemplo mad 2D}\rm{
 Set $q=2$, $r_1=3$ and $r_2=9$.  Let $M=M(D)$ be the matrix afforded by the set of $2$-orbits $D=Q(1,0)\cup Q(0,1)\cup   Q(1,3)\cup  Q(1,6)$.
 
 Following the construction given in the proof of proposition above we get $m_0=d^*(M)=3$ and $Ip(M_0)=\{(1,0),(2,0),(2,3),(2,6)\}$. We set 
 \[S=\bigcup_{(k,b)\in Ip(M)}supp(H_M(k,b))=\{(0,0),(0,3),(0,6)\}\]
 and construct $M_1$ by puting $0$ in the entries $a_{i,j}$ for which $(i,j)\in Q(0,0)\cup Q(0,3)$ (note that $(0,6) \in Q(0,3)$).
 
 One may see that $M_1\neq 0$; so we repeat the process. Now $d^*M_1=4$, $Ip(M_1)=\{(1,2),(2,2),(2,5),(2,8)\}$ and $m_1=\min\{m_0,d^*M_1\}=\min\{3,4\}=3$. In this case  $$S=\{(1,j) \tq j=2,5,8\}\cup \{(2,j) \tq j=1,2,4,5,7,8\}$$ and $\{Q(a_1,a_2)\tq (a_1,a_2)\in S\}=Q(1,2)\cup Q(2,2)$. This yields $M_2=0$.
 
 Thus, we obtain the sequences  $M>M_1$, $m_0=3\geq m_1=3$ and so $mad(M)=3=d^*M$. 
 
 We note that, following the methods in \cite{Camion} and \cite{Evans} we should compute the apparent distance of $15$ matrices.
}\end{example}

Next theorem is the main result of this section. Here, we give a general method that simplifies the computation of the minimum apparent distance of a hypermatrix.

\begin{theorem}\label{teodamvarias}
Let $s,q$, $r_1,\dots,r_s$ be positive integers, with $q$ a power of a prime number $p$, such that $p\nmid r_i$, for $i=1,\dots s$.  We set $I=\prod_{j=1}^s\Z_{r_j}$. Let $\Q_t$ be the set of all $q^t$-orbits modulo $(r_1, \dots, r_s)$, $\mu\in\{1,\dots, |\Q_t|-1\}$ and $\left\{Q_j\right\}_{j=1}^{\mu} $ a subset of $\Q_t$. Set
$D=\cup_{j=1}^\mu Q_j$ and let $M=M(D)$ be the $q^t$-orbits hypermatrix afforded by $D$. Then there exist two sequences: the first one is formed by pairwise disjoint sets of nonzero $q^t$-orbits hypermatrices 
$$\{M\}=\mathcal{M}_0 , \dots , \mathcal{M}_l \quad (\mathcal{M}_i\neq \emptyset,\;\forall i).$$
and the second one is formed by positive integers 
$$m_0 \geq \dots \geq m_l,$$ 
where $l\leq\mu$, each $L\in \mathcal{M}_i$ verifies that $L\leq M$ and $m_i \leq  \min\{d^\ast (L):L \in \mathcal{M}_i\}$, with $0 \leq i \leq  l$.  Moreover, these sequences verify that:

\begin{enumerate}
 \item If $0<i$ and $L \in \mathcal{M}_i$ then there exists $L' \in \mathcal{M}_{i-1}$ with $L<L'$.
 \item If $0<i$ and $P\leq M$ with $d^\ast(P) < m_{i-1}$ then $P \leq L$ for some $L \in \mathcal{M}_i$.
 \item  If $0\neq P \leq M$ then  $d^\ast (P) \geq m_l$.
\item If $l' \in \{0, \dots , l\}$ is the first element such that $m_{l'} =m_l$, then there are $P\leq M$, and $L \in \mathcal{M}_{l'}$, such that $P\leq L$ and $d^ \ast (P)=mad(M)=m_l$.
\end{enumerate}
\end{theorem}

\begin{proof}
As $\mu\leq |\Q_t|-1$ then $M\neq 0$. We proceed by induction on $s$. The case $s=1$ follows from the paragraph after Proposition~\ref{matrizdamvarias}  and $s=2$ is Proposition~\ref{teodam2}, taking $\M_i=\{M_i\}$, for $i=0,\dots, l$. Now suppose we know that our theorem holds for $2\leq s-1$. We shall prove it for $s$.

Let us recall that any hypercolumn of $M$, say $H_M(k,b)$, with $1\leq k \leq s$ and $b\in \Z_{r_k}$, may be viewed as a $q^{t|C_{q^t}(b)|}$-orbits hypermatrix of dimension $s-1$ (as we have seen at the beginning of this section), and, as such, by induction hypothesis, we may compute $mad(H_M(k,b))$ and obtain the sequences of this theorem for it.

Now we shall associate to each $q^t$-orbits hypermatrix $0\neq P\leq M$ a set, denoted by $\S(P)$, of $q^t$-orbits hypermatrices less than $P$, with respect to (\ref{matrixordering}). Recall that $Ip(P)$ denotes the set of involved pairs in the computation of $d^\ast (P)$. As we have already mentioned, we know that for each $(k,b)\in Ip(P)$ we may compute $mad (H_P(k,b))$ and construct sequences $\{H_P(k,b)\}=\H_0,\dots,\H_{l(k,b)}$ and $h_{0},\dots,h_{l(k,b)}$ satisfying properties \textit{(1)} to \textit{(4)} of this theorem. Now, for each $(k,b)\in Ip(P)$, set $h(k,b,i)=|\H_i|$, for $i\in \{0,\dots,l(k,b)\}=\N_{l(k,b)}$. We fix an arbitrary indexation on the elements of $\H_i$ with $\N_{h(k,b,i)}=\{0,\dots,h(k,b,i)-1\}$ and set $\gamma(k,b)=\left\{(u,v)\tq u\in \N_{l(k,b)},\;v\in\N_{h(k,b,u)}\right\}$.

By Lemma~\ref{subhipermat soporte maximo}, for each $(k,b)\in Ip(P)$ and $E=(u,v)\in \gamma(k,b)$ we may construct the $q^{t|C_{q^t}(b)|}$-orbits hypermatrix, that we call $(P,E)$, with maximum support, such that $H_{(P,E)}(k,b)$ is exactly the element of $\H_u$ with corresponding index $v$. We collect these hypermatrices in the sets
\begin{eqnarray*}
 R(P,k,b)&=&\left\{(P,E)\tq E\in\gamma(k,b)\right\}\setminus \{P\} \quad \text{and}\\ \S(P)&=&\bigcup_{(k,b)\in Ip(P)}\pega{3} R(P,k,b).
\end{eqnarray*}

Now we set $\M_0=\{M\}$. To construct $\M_1$ we shall collect first, by recursion, a sequence of disjoint sets of hypermatrices less than or equal to $M$, that we will denote $\T_0(M),\dots, \T_{n(M)}(M)$, where $n(M)\in \N$, that satisfies the following properties. For each $i=0,\dots,n(M)$, $\T_i(M)\neq \emptyset$ and for any $P\in \T_i(M)$, one has that $\S(P)\subseteq \T_{i+1}(M)$. In addition, $\T_{n(M)+1}(M)=\emptyset$ . To do this, we set $\T_0(M)=\{M\}$ and $\T_1(M)=\S(M)$. Now, once $\T_i(M)$ has been constructed, we set 
$\T_{i+1}(M)=\bigcup_{P\in \T_i(M)}\S(P)$. So that, if $\T_i(M)\neq \emptyset$ and $P\in \T_i(M)$ then $\S(P)\subseteq \T_{i+1}(M)$.

Note that, for each $j\in \{1,\dots,i+1\}$ and $P\in \T_j(M)$, there must exist $L\in \T_{j-1}(M)$ for which $P<L$ (strictly) and so, the construction of the sequence must stop. 

Let $n(M)\in \N$ be the first element for which  $\T{~\pega{2}_{n(M)+1}}(M)=\emptyset$. Now we set.
\begin{eqnarray*}
\T\left(\M_0\right)&=&\bigcup_{j=0}^{n(M)}\T_j(M)\\
m_0&=& \min\left\{d^\ast (N) \tq N\in \T(\M_0)\right\}\quad \text{and}\\
\eta_0&=& \left\{N \in \T (\M_0) \tq \S(N)=\emptyset\right\}.
\end{eqnarray*}

Let us remark that if $N\in \eta_0$ then for every $(k,b)\in Ip(N)$ we have that $l(k,b)=0$ and then $d^\ast (H_N(k,b))=mad(H_N(k,b))$; moreover, for any $(k,b)\in Ip(N)$, we have that $\{(N,E):E\in \gamma(k,b)\}=\{N\}$ and so $R(N,k,b)=\emptyset$; hence, for any $(k,b)\in Ip(N)$ and any hypermatrix $0\neq L<H_N(k,b)$ we have that $d^\ast (L)\geq d^\ast (H_N(k,b))$  as $(s-1)$-dimensional $q^{t|C_{q^t}(b)|}$-orbit hypermatrices. 

Another property of $\eta_0$ that we need is the following. If $P< M$ is a hypermatrix such that $d^\ast (P)< m_0$ then there exists $N\in \eta_0$ such that $P<N$. Indeed, first, if $\S(M)=\emptyset$ then $P<M\in \eta_0$ and we are done. So, suppose that $\S(M)\neq \emptyset$. By Lemma~\ref{submatricessubcolumnas} $H_P(k,b)<H_M(k,b)$ (strictly), for all $(k,b)\in Ip(M)$ and then  $P\leq L$ for some hypermatrix $L\in \S(M)=\T_1(M)$ and $d^\ast (P)<m_0\leq d^\ast (L)$. Again, if $\S(M)=\emptyset$ we are done; otherwise we may find, $L'\in \S(L)\subseteq\T_2(M)$ for which $P<L'$. We may continue the process until finding $N\in \eta_0$ with $P<N$, as desired.

Now suppose that  $P< M$, with $d^\ast (P)< m_0$ and $N\in \eta_0$ is such that $P<N$. Again $d^\ast (P)<m_0\leq d^\ast (N)$ and then by Lemma~\ref{submatricessubcolumnas} we have that $H_P(k,b)<H_N(k,b)$  for all $(k,b)\in Ip(N)$, which implies that $H_P(k,b)=0$ because of our remark two paragraphs above.

We are now ready to construct $\M_1$. Using Corollary~\ref{subhipermat de soporte maximo corolario}, for each $N\in \eta_0$, we define $L(N) < N$ as the hypermatrix of maximum support for which $H_{L(N)}(k,b)=0$, for all $(k,b)\in Ip(N)$. Then we define
	\[\M_1=\left\{L(N)\tq 0\neq L(N)\;\text{ and }\; N\in \eta_0\right\}.\]
Note that $\M_0\cap \M_1=\emptyset$, because of the construction of the $R(P,k,b)$'s.

If $\M_1 = \emptyset$ then $l=0$. If $\M_1\neq \emptyset$ we get a new element in our sequence; so we have $\M_0,\;\M_1$, and we have to check properties \textit{(1)} and \textit{(2)} of our theorem (properties \textit{(3)} and \textit{(4)} will be checked when we finish the construction of the sequences). Property \textit{(1)} is obvious, as $\M_0=\{M\}$. To check property \textit{(2)} we suppose that there is a hypermatrix $P\leq M$ with $d^\ast (P)<m_0$. As we have already seen, there exists $N\in \eta_0$ such that $P<N$, and for all $(k,b)\in Ip(N)$ we have that $H_P(k,b)=0$, so that $P\leq L(N)$, because $L(N)$ has maximum support.

Suppose we have constructed pairwise disjoint sets $\M_0,\M_1,\dots,\M_i$ with $\M_j\neq \emptyset$ for $j=1,\dots,i$ and a sequence $m_0\geq\dots \geq m_{i-1}$.  verifying properties \textit{(1)} and \textit{(2)} of our theorem and, moreover, if $P\in \M_{j+1}$ and $N\in\M_j$ is such that $P<N$ then $H_P(k,b)=0$ for all $(k,b)\in Ip(N)$ (in a similar way of the hypermatrices in $\M_1$). Suppose we also have constructed $\eta_0,\dots,\eta_{i-1}$ 

The induction step is analogous to the base step. For each $P\in \M_i$, we construct $\T_0(P),\dots,\T_{n(P)}(P)$ and collect

\begin{eqnarray*}
\T\left(\M_i\right)&=& \bigcup_{P\in\M_i} \bigcup_{j=0}^{n(P)} \T_j(P)\\
m_i&=& \min\left(\{d^\ast (P) \tq P\in \T\left(\M_i\right)\}\cup\left\{m_{i-1}\right\}\right), \\
\eta_i&=& \left\{N \in \T\left(\M_i\right) \tq \S(N)=\emptyset\right\}
\end{eqnarray*}
and
\[\M_{i+1}= \left\{L(N)\tq 0\neq L(N)\;\text{ and }\; N\in \eta_i\right\}.\]
Clearly, properties \textit{(1)} and \textit{(2)} of our theorem hold. If $\M_{i+1}=\emptyset$ we have finished our construction of sequences with $l=i$.

Now property \textit{(1)} guarantees that the process must stop; that is, there exists $l\in \N$ such that $\M_{l+1}=\emptyset$; in fact, each sequence could have at most $\mu$ elements.  So, suppose we have constructed $\M_0,\dots,\M_l$, $m_0\geq\dots \geq m_{l}$ such that $\M_l\neq \emptyset$ and $\M_{l+1}=\emptyset$. As $\T(M_l)\neq \emptyset$ then $\eta_l\neq \emptyset$ and then $m_l$ may be computed as above. So the sequences are completed. Now we have to check properties \textit{(1)} to \textit{(4)} of our theorem. As we have seen, properties \textit{(1)} and \textit{(2)} are immediate. To see property \textit{(3)}, we consider a nonzero $q^t$-orbits hypermatrix $P\leq M$ and suppose that $d^\ast (P)< m_l\leq m_0$. Then, $P\leq N$ for some $N\in \eta_l$, and, since $d^\ast (P)<d^\ast (N)$ it must happen that, as above, $H_P(k,b)=0$ for all $(k,b)\in Ip(N)$. However, $\M_{l+1}=\emptyset$, so that $P=0$.

Finally, let $l'\in \{0,\dots,l\}$ be the first element such that $m_{l'}=m_l$. As it has the minimum value then there must exist $P\in \T\left(\M_{l'}\right)$ such that $m_{l'}=d^\ast (P)$ and $P\leq L$ for some $L\in \M_{l'}$. The fact that $d^\ast (P)=mad(M)$ is obvious.
\end{proof}

In \cite[pp. 357-358]{BBS} the reader may find an explicit algorithm for the cases of two and three variables.

\begin{example}\rm{
 We are going to continue with the hypermatrix in Example~\ref{ejemplos dis apar}\textit{(3)}. We recall that $q=2$, $r_1=3$, $r_2=3$, $r_3=5$ and  $M=M(D)$ is the matrix afforded by the set of $2$-orbits 
  \begin{eqnarray*}
D&=& Q(0,0,0)\cup Q(1,0,0)\cup Q(0,1,0)\cup Q(0,0,1)\\
&& \cup Q(1,2,0)\cup Q(1,2,1)\cup Q(1,2,2)\cup Q(1,0,1)\\ 
&&\cup Q(0,1,1)\cup Q(1,0,2)\cup Q(0,1,2).
\end{eqnarray*}
Then $d^\ast (M)=6$ and $Ip(M)=\{(1,2),(2,2),(3,0),(3,1),$ $(3,2), (3,3),(3,4)\}$.

Now we shall compute $\mathcal{S}(M)$. To do this, we have to obtain the sequences for each $H_M(k,b)$ such that $(k,b) \in Ip(M)$ as it is described in the proof of the theorem above or in \cite[p. 358]{BBS}. This gives us $R(M,1,2)=R(M,2,2)=R(M,3,0)=\emptyset$ and  $R(M,3,1)=R(M,3,4)=\{B_1\}$, where $B_1$ is the hypermatrix such that $\D(B_1)=D\cup Q(1,1,2)$, and $R(M,3,2)=R(M,3,3)=\{B_2\}$, where $B_2$ is the hypermatrix such that $\D(B_2)=D\cup Q(1,1,1)$.

Hence $\S(M)=\left\{B_1,\,B_2\right\}$. One may check that, for $B_1$ we get $d^\ast\left(B_1\right)=12$, $Ip\left(B_1\right)=\left\{(1,2),\,(2,2)\right\}$ and $\S\left(B_1\right)=\emptyset$, and for $B_2$ we get $d^\ast\left(B_2\right)=18$, $Ip\left(B_2\right)=\left\{(1,2),\,(2,2)\right\}$ and $\S\left(B_2\right)=\emptyset$. Then $\T_0(M)=\{M\}$, $\T_1(M)=\left\{B_1,\,B_2 \right\}$ and $\T_2(M)=\emptyset$. So that $\T(M)=\left\{M,\, B_1,\, B_2\right\}$, $m_0=\min\{6,12,18\}=6$ and $\eta_0=\left\{B_1,\, B_2\right\}$.

Now we are going to construct $\M_1$. To do this, we have to consider, for each $N\in \eta_0$, the hypermatrix $L(N)<N$ having maximal support with respect to the property $H_{L(N)}(k,b)=0$ for all $(k,b)\in Ip(N)$. In our case, it happens that $L\left(B_1\right)=L\left(B_2\right)=0$; so that $\M_1=\emptyset$ and the process is finished. The sequences are $\{M\}$ and $m_0=6$. Hence $mad(M)=6$.

If we define the code $C$ in $A_2(3,3,5)$ with defining set $D_\alpha(C)=D$, for some $\alpha\in U$, then we have that $d^\ast_\alpha(C)=6$.
}\end{example}

\section{Multivariate BCH bound and BCH code}

This section is devoted to generalize the notion of BCH codes to the multivariate case. We also study the extension of  most of the clasical results about this codes. To do this, we first present a generalization of the notion of BCH bound and BCH code, based on the results that we have seen in the previous sections. Then we shall show that most of the classical results for BCH codes can be generalized to our setting.

We keep all notation from the preceding sections; that is, $s$, $t$, $q$ and $r_1,\dots,r_s$ are positive integers, with $q$ a power of a prime number $p$, such that $p\nmid r_i$, for $i=1,\dots s$  and $A_q(r_1,\dots,r_s)$ is the quotient ring $\F_q[x_1,\dots,x_s]/\left\langle x_1^{r_1}-1,\dots,x_s^{r_s}-1\right\rangle$. Also,  $I=\prod_{j=1}^s\Z_{r_j}$, $\Q_t$ is the set of all $q^t$-orbits modulo $(r_1, \dots, r_s)$ and if  $D$ is a union of $q^t$-orbits then  $M(D)$ denotes the $q$-orbits hypermatrix afforded by $D$. We recall that  $U_{r_i}$ denotes the set of all $r_i$-th primitive roots of unity, for each $i=1,\dots, s$ and we define $U=\{(\alpha_1, \dots ,\alpha_s) \tq \alpha_i \in U_{r_i}\}$.

Our first task is to extend the notion of BCH bound. Next result is the first step in order to establish this extension.

\begin{lemma}\label{designando distancia una cara}
Let $\delta\in\Z$ be such that $\delta\geq 2$ and consider $\alpha \in U$. Let $0 \neq C$ be an abelian code in $A_q(r_1,\dots,r_s)$, $\D_{\alpha}(C)$ its defining set and $M=M\left(\D_{\alpha}(C)\right)$. If there are an element $k\in \{1,\dots,s\}$ and a list of $(\delta-1)$-consecutive integers modulo $r_k$, $\{j_0,\dots,j_{\delta-2}\}$ such that  $H_M(k,j_i)=0$, with $i=0,\dots,\delta-2$, then $d^\ast (C)\geq d_{\alpha}^\ast (C)\geq \delta$.
\end{lemma}

\begin{proof}
We know that $mad(M)=d_{\alpha}^\ast (C)\leq d^\ast(C)\leq d(C)$; so, we have to see that $\delta \leq mad(M)$. Let $P\leq M$ be a hypermatrix such that $d^\ast(P)= mad(M)$. Since $P\leq M$ then $H_P(k,j_i)=0$, for all $i=0,\dots,\delta-2$.  Let $\{j'_0,\dots,j'_{\delta'-2}\}$ be the biggest list of consecutive integers modulo $r_k$ containing $\{j_0,\dots,j_{\delta-2}\}$ which $H_P(k,j'_i)=0$, for all $i=0,\dots,\delta'-2$. Since $C\neq 0$ then $H_P(k,j^{\prime\;-}_0)\neq 0$ and so $\omega(k,j^{\prime\;-}_0)=\delta'\geq \delta$. Hence $\delta \leq \delta' \leq d^\ast(P)$.
\end{proof}

Now we deal with the general case.
\begin{theorem}\textbf{(Multivariate BCH bound)}\label{teo distancia designada}
Let $s,q$, $r_1,\dots,r_s$ be positive integers, with $q$ a power of a prime number $p$, such that $p\nmid r_i$, for $i=1,\dots s$ and $\alpha \in U$.  We set $I=\prod_{j=1}^s\Z_{r_j}$. Let $C$ be a nonzero abelian code in $A_q(r_1,\dots,r_s)$ with defining set $\D_{\alpha}(C)$ and $M$ the $q$-orbits hypermatrix afforded by $\D_{\alpha}(C)$. Suppose that there exist a subset $\gamma\subseteq \{1,\dots,s\}$ and a list of integers $\left\{\delta_k\geq 2 \tq k\in\gamma\right\}$ satisfying the following property:  for each $k\in\gamma$, the hypermatrix $M$ has zero hypercolumns $H_M(k,j_{(k,0)}), \dots , H_M(k,j_{(k,\delta_k-2)})$, where $\{j_{(k,0)}, \dots , j_{(k,\delta_k-2)}\}$ is a list of consecutive integers modulo $r_k$. Then $d_{\alpha}^\ast (C)\geq \prod_{k\in\gamma}\delta_k$. Hence, $d^\ast (C)\geq \prod_{k\in\gamma}\delta_k$.
\end{theorem}

\begin{proof}
We proceed by induction on $|\gamma|=l$. The case $l=1$ is covered by the lemma above. So, suppose that, our theorem is true for $l-1$ and let $C$ be a nonzero abelian code in $A_q(r_1,\dots,r_s)$ with defining set $\D_\alpha(C)$ and $M$ the $q$-orbits hypermatrix afforded by $\D_\alpha(C)$.  Let $P\leq M$ be a hypermatrix such that $d^\ast(P)= mad(M)$. Since $P\leq M$ then $H_P(k,j_{(k,i)})=0$ for $k\in \gamma$ and $i=0,\dots,\delta_k-2$. For each $k \in \gamma$, we set $N_k=H_P(k,b_k)$ the nonzero hypercolumn of $P$ such that $\left\{H_P(k,j_{(k,i)})\right\}_{i=0}^{\delta_k-2}\subseteq CH_P(k,b_k)$ (see Definition~\ref{hipercolumnas cero adyacentes}). By the definition of apparent distance of a hypermatrix we have that $d^\ast (N_k)\cdot \delta_k\leq d^\ast (P)$ for any $k\in\gamma$.  Now, we know that $H_{N_k}(u,j_{(u,i)})=0$, with $u\in \gamma\setminus \{ k\}$ and $i=0,\dots,\delta_u-2$.  Let $C_{N_k}$ be the abelian code in $A_q(r_1,\dots,r_s)$  such that $\D(N_k)=\D_\alpha(C_{N
 _k})$ where we are considering $N_k$ as a hypermatrix indexed by $\prod_{j\neq k}\Z_{r_j}$. Then $M(\D_{\alpha}(C_{N_k}))=N_k$.  By induction hypothesis $d^\ast (N_k) \geq mad(N_k)=d^\ast \left(C_{N_k}\right) \geq  \displaystyle\prod
 _{u\in \gamma\setminus \{ k\}} \delta_u$. Hence $ d^\ast(P)\geq \prod_{k\in\gamma}\delta_k$ and we are done.
\end{proof}

Let us reformulate last theorem in terms of lists of positive integers. We recall that for any element $b\in \Z$ and any positive integer $r$, we denote by $\overline{b}$ the canonical representative of $b$ in $ \Z_{r}$.

\begin{corollary}
Let $\gamma\subseteq \{1,\dots,s\}$ be a set, and let $\delta=\{\delta_k\geq 2 \tq k\in \gamma\}$ and $b=\{b_k\geq 0\tq k\in \gamma\}$ be lists of integers. For each $k \in \gamma$ consider the list of consecutive integers modulo $r_k$,  $J_k= \{\overline{b_k},\dots,\overline{b_k+\delta_k-2}\}$ and set $A_k=\{\b{i}\in I \tq \b{i}(k)\in J_k\}$. If $C$ is a nonzero abelian code satisfying $\cup_{k=1}^sA_k \subseteq \D_{\alpha}(C)$, for some $\alpha \in U$, then $d^\ast (C)\geq \prod_{k\in\gamma}\delta_k$.
\end{corollary}

\begin{proof}
Immediate from the theorem above by taking $j_{(k,0)}=\overline{b_k}$ until $j_{(k,\delta_k-2)}=\overline{b_k+\delta_k-2}$. 
\end{proof}

\begin{example}\label{cotas bch}\rm{
 Let $C_1$ and $C_2$ be the abelian codes in $A_2(5,7)$, with defining sets $D_1=D_{\alpha}(C_1)=Q(0,1)\cup Q(1,1)$ and $D_2=D_{\alpha}(C_2)=D_1\cup Q(0,0)\cup Q(0,3)$, with respect to some $\alpha \in U$. We set $M_1= M(D_1)$ and $M_2=M(D_2)$.
 
 We are going to apply Lemma~\ref{designando distancia una cara} to $C_1$. A simple inspection shows us that $H_{M_1}(2,1)=H_{M_1}(2,2)=H_{M_1}(2,4)=0$ so that by taking $k=2$, $j_0=1$ and $j_1=2$ we get $\delta=3$. Therefore $d^\ast (C_1)\geq 3$.
 
 Now we apply Theorem~\ref{teo distancia designada} to $C_2$. Again, a simple inspection shows us that $H_{M_2}(1,0)=H_{M_2}(2,1)=H_{M_2}(2,2)=H_{M_2}(2,4)=0$. In this case, we take $\gamma=\{1,2\}$, $j_{1,0}=0$, $j_{2,0}=1$ and $j_{2,1}=2$. So that $\delta_1=2$, $\delta_2=3$ and hence $d^\ast (C_2) \geq 6$. 
}\end{example}

We are ready to present a new notion of multivariate BCH code. We recall that $I(k,l)=\{\mathbf i \in I: \mathbf{i}(k)=l\}$

\begin{definition}\textbf{(Multivariate BCH code)}\label{codigo bch multivariable}
Let $s,q$, $r_1,\dots,r_s$, $I$ be as above. Let $\gamma \subseteq \{1,\dots,s\}$ and $\delta=\{r_k\geq \delta_k\geq 2\tq k\in\gamma\}$. An abelian code $C$ in $A_q(r_1,\dots,r_s)$ is a multivariate BCH code of designed distance $\delta$ if there exists a list of positive integers $b=\{b_k\tq k\in\gamma\}$ such that
\[\D_{\alpha}(C)=\bigcup_{k\in\gamma} \bigcup_{l=0}^{\delta_k-2}\bigcup_{\b{i}\in I(k,\overline{b_k+l})}Q(\b{i})\]
for some $\alpha \in U$, where $\{\overline{b_k},\dots,\overline{b_k+\delta_k-2}\}$ is a list of consecutive integers modulo $r_k$. We denote $C=B_q(\alpha,\gamma,\delta,b)$, as usual.
\end{definition}

As a direct consequence of Theorem \ref{teo distancia designada}  we have that $d^\ast \left(B_q(\alpha,\gamma,\delta,b)\right) \geq \prod_{k\in\gamma}\delta_k$. 

\begin{example}\rm{
In Example~\ref{cotas bch}, the code $C_1$ is a multivariate BCH code, $C_1=B_2\left(\alpha,\{2\},\{3\},\{1\}\right)$, while $C_2=B_2\left(\alpha,\{1,2\},\{2,3\},\{0,1\}\right)$.

Let us show an example of a multivariate BCH code in $A_2(3,5,5)$. Let $C_3$ be the abelian code with defining set $D_3=\D_\alpha(C_3)=Q(0,0,0)\cup Q(0,0,1)\cup Q(0,1,0)\cup Q(1,0,0)\cup Q(1,0,1)\cup Q(1,0,2)\cup Q(1,1,0)\cup Q(1,2,0)$, with respect to some $\alpha\in U$. Set $M_3=M(D_3)$. In this case $H_{M_3}(2,0)=H_{M_3}(3,0)=0$ so we may take $\gamma=\{2,3\}$, $\delta=\{2,2\}$ and $b=\{0,0\}$ to conclude that $C_3=B_2\left(\alpha,\{1,2\},\{2,2\},\{0,0\}\right)$.
}\end{example}

From now on, we shall extend the basic properties of BCH codes to the multivariate case. The following property is immediate.

\begin{corollary}
 Let $B_q(\alpha,\gamma,\delta,b)$ be a multivariate BCH code. For each $k\in \gamma$, set $J_k=\left\{\overline{b_k},\dots,\right.$ $\left.\overline{b_k+\delta_k-2}\right\}$ and $A_k=\{\b{i}\in I \tq \b{i}(k)\in J_k\}$.  If $C$ is an abelian code in $A_q(r_1,\dots,r_s)$ such that $\cup_{k=1}^sA_k \subseteq \D_{\alpha}(C)$ then $\dim C\leq \dim B_q(\alpha,\gamma,\delta,b)$.
\end{corollary}

It is known (see \cite[Theorem 10, p. 203]{sloane}) that any (univariate) BCH code $B=B_q(\alpha,\delta,b)$ in $A_q(r_1,\dots,r_s)$ verifies that $d(B)\geq \delta$ and $\dim(B)\geq r-m(\delta-1)$, where $m=\mathcal O_r(q)$. In the multivariate case we have the following result.

\begin{theorem}\label{dimBCH mayor igual orbitas por hipermatrices}
 Let $s,q$, $r_1,\dots,r_s$ be positive integers, with $q$ a power of a prime number $p$, such that $p\nmid r_i$, for $i=1,\dots s$.  We set $I=\prod_{j=1}^s\Z_{r_j}$.  Let $B_q(\alpha,\gamma,\delta,b)$ be a multivariate BCH code with $\delta=\{\delta_k \geq 2\tq k\in \gamma\}$ and $b=\{b_k\geq 0\tq k\in \gamma\}$. %as in Definition~\ref{codigo bch multivariable}.
 Then $\dim_{\F_q}B_q(\alpha,\gamma,\delta,b)\geq \prod_{j =1}^s r_j-m\left(\sum_{k\in\gamma}\left( \left(\delta_k-1 \right) \prod_{\doble{j=1}{j\neq k}}^s r_j \right)   \right),$
where $m= \lcm\left\{\mathcal{O}_{r_k}(q)\right\}_{k=1}^s$.
\end{theorem}

\begin{proof}
Set $C=B_q(\gamma,\delta,b)$. By definition we have that \[D_\alpha(C)=\bigcup_{k\in\gamma} \bigcup_{l=0}^{\delta_k-2}\bigcup_{\mathbf{i} \in I(k,\overline{b_k+l})} Q(\mathbf{i}).\]
Clearly $|I(k,h)|=\prod_{\doble{j=1}{j\neq k}}^s r_j$ for all $h\in \Z_{r_k}$, and for any $\b{i}\in I$ we have that  $\left|Q(\b{i})\right| \leq \lcm\left\{\mathcal{O}_{r_k}(q)\right\}_{k=1}^s=m$, so that
$$\begin{array}{lll}
\dim(C)&=& \prod_{j =1}^s r_j-|\D_\alpha(C)|\\
&\geq&  \prod_{j =1}^s r_j-m\left(\sum_{k\in \gamma}\left( \left(\delta_k-1 \right) \prod_{\doble{j=1}{j\neq k}}^s r_j \right)   \right).
\end{array}
$$

\end{proof}

\section{Applications}

\subsection{Multiplying dimension in abelian codes}

We shall construct abelian codes starting from BCH (univariate) codes with designed distance $\delta\in \N$. We keep all notation from the preceding sections.

\begin{lemma}\label{lema dimension multiplicada}
 Let $D$ be a union of $q$-orbits modulo $(r_1,r_2)$ and consider the $q$-orbits matrix $M=M(D)$. The following conditions on $M$ are equivalent:
\begin{enumerate}
 \item \label{condicion 1 lema dimension multi} Each column $H_M(2,j)$ verifies that either $H_M(2,j)=0$ or all of its entries have constant value $1$.
\item \label{condicion 2 lema dimension multi} For all $(i,j)\in\Z_{r_1}\times\Z_{r_2}$, it happens that $(i,j) \in D$ if and only if $(x,j) \in D$ for all $x \in \Z_{r_1}$. 
\end{enumerate}
\end{lemma}

\begin{proof}
The result comes immediately from the definition of (hyper)matrix afforded by $D$; that is, for any $a_{ij}\in M$, $a_{ij}=0$ if and only if $(i,j)\in D$ and, otherwise, $a_{ij}=1$.
\end{proof}

As the reader may see, an analogous result may be obtained by replacing $r_2$ by $r_1$. For our next theorem we recall that, associated to the computation of the apparent distance of a hypermatrix, we defined the set  of optimized roots of $C$ as $\mathcal R(C)=\{\beta \in U \tq d^\ast (C)=d_{\beta}^\ast (C))\}$. We also keep the notation in Remak~\ref{distintas raices}.

\begin{theorem}\label{teorema dimension multiplicada}
 Let $n$ and $r$ be positive integers such that $\gcd(q,nr)=1$.  Let $C$ be a nonzero cyclic code in $A_q(r)$ with $d^\ast (C)= \delta >1$ and $\alpha=(\alpha_1,\alpha_2) \in U_n \times \mathcal{R}(C)$. Then, the abelian code $C_n$ in $A_q(n,r)$ with defining set $\D_{\alpha}(C_n)=\Z_n\times \D_{\alpha_2}(C)$ verifies that $d^\ast (C_n) = \delta$ and $\dim_{\F_q}(C_n)=n\dim_{\F_q}(C)$.
\end{theorem}

\begin{proof}
Consider $\beta = (\beta_1, \beta_2) \in U_n\times U_r$ and let $C_n$ the abelian code such that $\D_\beta(C_n) = \Z_n\times \D_{\beta_2}(C)$. It is clear that $\D_\beta(C_n)$ satisfies the condition~\textit{(\ref{condicion 2 lema dimension multi})} of Lemma \ref{lema dimension multiplicada}; so, the $q$-orbits matrix afforded by $\D_\beta(C_n)$, $N = M(\D_\beta(C_n))$, verifies the condition~\textit{(\ref{condicion 1 lema dimension multi})} of that lemma. If  $M = (a_j),j\in \Z_r,$ is the $q$-orbits vector afforded by $\D_{\beta_2}(C)$ then $H_N(2, j) = 0$ if and only if $a_j = 0$. Hence, and since $d^\ast(C) = \delta$, for all nonzero row of $N$, $H_N(1, b)$, we have that $\omega_N(1, b) = 0$ and $d^ \ast (H_N(1, b)) = d^\ast (M)\leq \delta$. So $d^\ast_1(N) = \max\{ (\omega_N(1, b)+1)d^\ast (H_N(1, b))\} = d^\ast (M)\leq \delta$. Moreover, the equality is reached when $\beta_2\in \mathcal R(C)$.
 
Now, for all nonzero column of $N$, $H_N(2, b)$, we have that $\omega_N(2, b)< d^ \ast (M)\leq \delta$ and $d^\ast (H_N(2, b)) = 1$, hence $d^\ast_2(N) = \max\{ (\omega_N(2, b)+1)d^\ast (H_M(2, b))\} \leq d^\ast (M)\leq \delta$. Therefore, $d^\ast (N) = d^\ast (M) \leq \delta$ and from Proposition~\ref{matrizdamvarias}, $d^\ast_\beta (C_n) = mad(N) = d^\ast (N)\leq \delta$. The equality is reached if $\beta_2\in \mathcal R(C)$.
 
Finally, since $dim_{\F_q} (C_n) = |supp(N)|$, we have that $dim_{\F_q} (C_n) = n dim_{\F_q} (C)$.
\end{proof}

\begin{example}\rm{
 Set $q=2$, $r=55$, $n=3$, $\alpha=(\alpha_1,\alpha_2) \in U_{3}\times U_{55}$ and let  $C$ be the cyclic code in $A_2(55)$ with defining set with respect to $\alpha_2$, $D=\D_{\alpha_2}(C)=C_2(1)\cup C_2(5)$. Set $M=M(D)$. A simple inspection on $M$ shows us that $C$ is a BCH code with parameters $C=B_2(\alpha_2,7,13)$ and dimension 25. By applying Theorem~\ref{teorema dimension multiplicada} we construct the new bivariate code $C_1$ with defining set $\D(C_1)=\Z_3 \times D$. So that $d^\ast (C_3)=7$, $\dim_{\F_2}(C_3)=75$ and its length is $165$. In fact $C_1= B_2(\alpha,\{2\},\{7\},\{13\})$, by Lemma~\ref{designando distancia una cara}.
}\end{example}

In order to multiply dimension in multivariate Reed Solomon codes we have the following result.

\begin{proposition}\label{multiplicar dimension una cara}
Let  $B_q(\alpha,\gamma,\delta,b)$ be a multivariate BCH code with $\gamma=\{k\}$, $\delta=\{\delta_k\}$ and $b=\{b_k\}$, for some $k\in \{1,\dots,s\}$. If $r_k=q-1$ then $d_{\alpha}^\ast\left(B_q(\alpha,\gamma,\delta,b)\right)=\delta_k$ and $\dim_{\F_q}(B_q(\alpha,\gamma,\delta,b))=\left(r_k-\delta_k+1\right)\prod_{\doble{j=1}{j\neq k}}^s r_j$.
\end{proposition}

\begin{proof}
Since $r_k=q-1$ we have that $l q \equiv l \mod r_k$, hence $Q(\b{i}) \subseteq I(k,l)$  for all $\b{i} \in I(k,l)$.  So, $$\D_\alpha(C)=\bigcup_{l=\overline{b_k}}^{\overline{b_k+\delta_k-2}} I(k,l).$$ 
By following the proof of Theorem~\ref{dimBCH mayor igual orbitas por hipermatrices} in the case $\gamma=\{k\}$, for some $k\in \{1,\dots,s\}$ we have that $\dim_{\F_q}(C)=\prod_{j=1}^s r_j-|\D_\alpha(C)|=\left(r_k-\delta_k+1\right)\prod_{\doble{j=1}{j\neq k}}^s r_j$. 

Note that if $M=M(\D_\alpha(C))$ and $l \in \{\overline{b_k}, \dots, \overline{b_k+\delta_k-2}\}$ then $H_M(k,l)=0$, otherwise, all of the entries of $H_M(k,l)$ have constant value 1. Therefore, $\omega_M(k,\overline{b_k}-1)=\delta_k-1$ and $d^\ast (H_M(k,\overline{b_k}-1))=1$, then $d_k^\ast (M)=\delta_k$.\\ 
We claim that $d^\ast (M)=d_k^\ast (M)$ and we shall prove it by induction on $s$. The case $s=1$ is trivial.  Assume that it is true for $s-1$ ($s\geq 2$) and consider $j \in \{1, \dots, s\}\setminus k$. Then, for all $l\in \Z_{r_j}$, the hypercolumn $H_M(j,l)$ may be viewed as an $(s-1)$-dimensional $q^{|C_q(l)|}$-orbits  hypermatrix indexed by $J=\prod_{\doble{i=1}{i\neq j}}^s \Z_{r_i}$, say $H_M(j,l)=\left(a_{\b{i}} \right)_{\b{i}\in J}$, where
\begin{equation*}
 a_{\b{i}}=\begin{cases}
   0 & \text{if}\; \b{i}(k) \in \{\overline{b_k}, \dots, \overline{b_k+\delta_k-2}\} \\
	1 & \text{otherwise.}  
  \end{cases}
\end{equation*}
Note that for all $l \in \Z_{r_j}$ it happens that $H_M(j,l) \neq 0$, which implies that $\omega_M(j,l)=0$.  In addition, $H_M(j,l)$ may be viewed as an $(s-1)$-dimensional hypermatrix afforded by $$\bigcup_{l=\overline{b_k}}^{\overline{b_k+\delta_k-2}}\, \left \{\b{i} \in \prod_{\doble{m=1}{m \neq j}}^s \Z_{r_m}\tq \b{i}(k)=l\right\}.$$ By induction hypothesis we have that  $d^\ast (H_M(j,l))=d_k^\ast (H_M(j,l))= \delta_k$. So, $d_j^\ast (M)=\delta_k$ for all $j \in \{1, \dots, s\}\setminus k$ and therefore, $d^\ast (M)=\max \{d_i^\ast (M) \tq i=1,\dots s\}=d_k^\ast (M)=\delta_k$, as desired.  Now, from Proposition~\ref{matrizdamvarias}, $mad(M)=d^\ast (M)=\delta_k$. Then,  $d_\alpha^\ast (B_q(\gamma,\delta,b))=\delta_k$ .
\end{proof}

The proposition above is applicable to the codes that we obtain by using the construction given in Theorem~\ref{teorema dimension multiplicada}, when we start from Reed-Solomon codes.

\begin{corollary}
Let $R=B_q(\alpha,\delta,b)$ be a Reed-Solomon code of length $r$. Then, for each positive integer $n$ and any $\alpha'\in U_n$, there exists a multivariate BCH code, $C=B_q\left(\{\alpha',\alpha\},\{2\},\{\delta\},\{b\}\right)$, such that $\dim(C)=\left(r-\delta+1\right)n= n \cdot \dim(R)$ and $d_\alpha^\ast(C)=\delta$.
\end{corollary}

\subsection{Designing maximum dimensional abelian codes (HD codes) for prescribed bounds}

In this section we will give another application of our techniques. Given an ambient space (and then a fixed length), we will design abelian codes with the highest dimension with respect to a fixed value for their apparent distance (HD codes, for short). As the reader will see our ideas are based in the consideration of the distribution of the $q$-orbits on the indexes of hypermatrices. We have used the GAP program to compute the minimum distance of some codes.  We begin with an example of codes of length 35.

\begin{example}\rm{
  We shall design HD codes in $\F_2^{35}$. We begin by constructing HD cyclic codes, that is, BCH codes. Observe the distribution of $2$-orbits ($2$-cyclotomic cosets in this case) in a $1\times 35$ vector:
\vspace{-.4cm}

\begin{footnotesize}\begin{eqnarray*}
 [ Q(0),Q(1),Q(1),Q(3),Q(1),Q(5),Q(3),Q(7),Q(1),Q(1), Q(5),\\
  Q(1), Q(3),Q(3),Q(7),Q(15),Q(1),Q(3),Q(1),Q(3),Q(5),Q(7),\\
 Q(1),Q(1), Q(3),Q(15),Q(3),Q(3),Q(7),Q(1),Q(15),\\
 Q(3),Q(1),Q(3),Q(3)] 
\end{eqnarray*}\end{footnotesize}
and $K(35)=\{1,3\}$ (see Remark~\ref{distintas raices}).\\

Then, one may see that the vectors afforded by $D_1=Q(1)\cup Q(5)$ or $D_2=Q(3)\cup Q(15)$ verify that $d^\ast (M(D_1))=d^\ast (M(D_2))=5$. Consider $\alpha \in U_{35}$ and let $C_1$ the cyclic code such that $\D_\alpha(C_1)=D_1$.  Observe that $D_2=3 \cdot D_1$ (see Remark~\ref{distintas raices}).  One may check that $d^\ast (C_1)=5$. Clearly, $C_1$ is a BCH code of dimension $20$ (so that, it is a HD code) and one may check that its minimum distance is $d(C_1)=6$. 

Now we want to fix a higher apparent distance. To do this, we extend the defining set of $C_1$ to $D_3=D_1\cup Q(7)$ and note that $3\cdot D_3=D_2\cup Q(7)$. Set $D_4=3\cdot D_3$. Then, one may check that the cyclic code $C_3$ with $\D_\alpha(C_3)=D_3$ has apparent distance $6$ and dimension $16$. One may also check that $d(C_3)=7$. In fact, this code have the highest dimension with respect to this bound; that is, it is a BCH code. Finally, a simple inspection shows us that any BCH code of length 35 with designed distance greater than $6$ must contain the cyclotomic cosets $Q(1)\cup Q(3)$ and hence its dimension must be less than $11$.

Now we want to construct bivariate HD codes in $A_2(5,7)$ (so that all of them will have the same length). Now, the distribution of $2$-orbits in a $5\times 7$ matrix is as follows

\begin{footnotesize}
\[  \left(\begin{array}{l@{ }l@{ }l@{ }l@{ }l@{ }l@{ }l}
    Q(0,0)& Q(0,1)&Q(0,1)&Q(0,3)&Q(0,1)&Q(0,3)&Q(0,3)\\
     Q(1,0)& Q(1,1)&Q(1,1)&Q(1,3)&Q(1,1)&Q(1,3)&Q(1,3)\\
     Q(1,0)& Q(1,1)&Q(1,1)&Q(1,3)&Q(1,1)&Q(1,3)&Q(1,3)\\
     Q(1,0)& Q(1,1)&Q(1,1)&Q(1,3)&Q(1,1)&Q(1,3)&Q(1,3)\\
     Q(0,0)& Q(1,1)&Q(1,1)&Q(1,3)&Q(1,1)&Q(1,3)&Q(1,3)
  \end{array}\right)\]
\end{footnotesize}
and $K(5,7)=\{(1,1),(1,3)\}$.

Setting $D_5=Q(0,0)\cup Q(1,0)\cup Q(0,3)$ we have that $mad (M(D_5))=4$. Consider $\alpha \in U_5 \times U_7$ and the abelian code such that $\D_\alpha(C_5)=D_5$.  Note that if $D=(1,3)\cdot D_5$ then $mad(M(D))=4$. It is easy to check that $C_5$  is a code of dimension $27$ and $d^\ast (C_5)=4$. By using the GAP program we obtain that $d(C_5)=4$. It is interesting to note that there are no cyclic codes with this parameters (see the remark below). However, $C_5$ is not a HD code with apparent distance $4$. In fact, if we consider the code $C'_5$ such that $\D_\alpha(C'_5)=D_5\setminus Q(0,0)$ we have that $C'_5$ is a HD code with apparent distance $d^\ast (C'_5)=4$ and its dimension is $28$. The reader may check that the distribution of the $q$-orbits in the matrix above shows us that any abelian code in $A_2(5,7)$ with dimension greater than 28 has apparent distance less than 4.  

Now we want to fix a higher apparent distance. To do this, we consider now $D_6=Q(0,1) \cup Q(0,3) \cup Q(1,3)$. The matrix afforded by $D_6$ has $mad (M(D_6))=6$ and the abelian code, $C_6$, such that $\D_\alpha(C_6)=D_6$, is a HD code with apparent distance $6$, dimension $17$ and minimum distance $6$. Finally, set $D_7=Q(0,0) \cup Q(1,0) \cup Q(0,1) \cup Q(0,3) \cup Q(1,3)$. The matrix afforded by $D_7$ has $mad (M(D_7))=8$ and the abelian code, $C_7$, such that $\D_\alpha(C_7)=D_7$ is a HD code with $d^\ast (C_7)=8$, dimension $13$ and minimum distance $8$.
}\end{example}

\begin{remark}\rm{
  In the example above, one may check that for any binary cyclic code $C$ of length $35$, it happens that if $\dim C\geq 25$ then its apparent distance verifies that $d^\ast(C)\leq 3$, because $\D_\alpha(C)$ cannot include neither of $Q(1)$ nor $Q(3)$. Then, even the abelian code $C_5$ may be seen as a cyclic code (via the Chinese Remainder Theorem) its apparent distance, computed as a cyclic code, will never be more than $3$.
}\end{remark}

\section{Conclusion}
We developed  an algorithm to computing the minimum apparent distance of a hypermatrix which noticeably reduces the number of involved operations in the computation of the apparent distance of an abelian code, with respect to the methods proposed in \cite{Camion} and \cite{Evans}.  In the two dimensional case the number of computations is of \textit{linear} order instead of exponential order.  Our techniques allowed us to give a notion of BCH multivariate bound and code, respectively. Moreover, we construct abelian codes from cyclic codes preserving their apparent distance and multiplying their dimension, and, given an ambient space, we design abelian codes with the highest dimension with respect to a fixed value for their apparent distance.


% \section*{Acknowledgment}
% 
% 
% The authors would like to thank professors C\'esar Polcino and Raul Ferraz for helpful comments and suggestions.

\end{document}